\documentclass[11pt,fleqn]{article}
\RequirePackage[T1]{fontenc}
\RequirePackage{amsthm,amsmath}
\RequirePackage[round]{natbib}

\usepackage{float}
\usepackage[english]{babel}
\RequirePackage[colorlinks,citecolor=blue,urlcolor=blue]{hyperref}
\usepackage{enumerate}
\usepackage{amsmath,amsfonts,amssymb,euscript,mathrsfs}
\usepackage{fullpage}
\usepackage{graphicx}
\usepackage{colortbl} 
\usepackage{xcolor}
\usepackage{subcaption}
\usepackage{dsfont}
\usepackage{color}   
\makeatletter
\def\captionof#1#2{{\def\@captype{#1}#2}}
\makeatother
\def\1{\mbox{\bf 1}}
\def\R{\mathbb{R}}

\def\N{\mathbb{N}}
\def\P{\mathbb{P}}
\def\E{\mathbb{E}}
\def\L{\mathbb{L}}

\def\R{\mathbb{R}}

\def\v{\mbox{Var}}

\newtheorem{theo}{Theorem}
\newtheorem{lem}{Lemma}

\newtheorem{cor}{Corollary}
\newtheorem{Def/Prop}{Definition-Proposition}

\newcommand{\blue}{\color{blue}}   

\newcounter{exos}
\renewcommand\theexos{\arabic{exos}}

\newcounter{prob}
\renewcommand\theprob{\arabic{prob}}

\begin{document}

\title{Inferring the parameters of Taylor's law in ecology}
 \author{ Lionel Truquet
(\footnote{\texttt{Lionel.TRUQUET@ensai.fr}.  Univ Rennes, Ensai, CNRS, CREST -- UMR 9194, F-35000 Rennes, France}\ ),
Joel E. Cohen (\footnote{
 \texttt{cohen@rockefeller.edu}. 
 Laboratory of Populations, The Rockefeller University, New York, NY 10065, USA;\\
also affiliated with Department of Statistics and Earth Institute, Columbia University, New York, NY 10027; and Department of Statistics, University of Chicago, Chicago, IL 60637, USA,  ORCID 0000-0002-9746-6725}
\ ),
 Paul Doukhan
(\footnote{\texttt{doukhan@cyu.fr}. CY Cergy Paris University,  AGM Mathematics, 2 Bd. Adolphe Chauvin, 95000 Cergy-Pontoise, AGM-UMR 8080}\ ) 
 }

\maketitle

\begin{abstract}
\noindent
Taylor's power law (TL) or fluctuation scaling has been verified empirically for the abundances of many species, human and non-human, and in many other fields including physics, meteorology, computer science, and finance. TL asserts that the variance is directly proportional to a power of the mean, exactly for population moments and, whether or not population moments exist, approximately for sample moments. In many papers, linear regression of log variance as a function of log mean is used to estimate TL's parameters. We provide some statistical guarantees with large-sample asymptotics for this kind of inference under general conditions, and we derive confidence intervals for the parameters. In many ecological applications, the means and variances are estimated over time or across space from arrays of abundance data collected at different locations and time points. When the ratio between the time-series length and the number of spatial points converges to a constant as both become large, the usual normalized statistics are asymptotically biased. We provide a bias correction to get correct confidence intervals. TL, widely studied in multiple sciences, is a source of challenging new statistical problems in a nonstationary spatiotemporal framework. We illustrate our results with both simulated and real data sets.
\end{abstract}


\vspace*{1.0cm}

\section{Introduction}

Taylor's law (TL) is an important family of relations between the mean and the variance of some probability distributions. 
TL is well approximated empirically in many data sets in the natural and social sciences, finance, and technology. 
Since \citet{taylor1961aggregation} publicized this relationship, 
which was discovered by others at least 20 years earlier, such as \citet{Bl}, 
thousands of papers have studied TL
for species abundance data in ecology and more generally to assess time-dependent 
or space-dependent changes of population distributions. 
\citet{taylor2019} and the references therein give a broad introduction to TL.

The aim of this paper is to study the most basic version of this law, 
Taylor's power law, as defined by \eqref{main} at the beginning of the next section.
This law, which is based on a simple power relation between the variance and the mean of the marginal distribution, has an important interpretation in term of aggregation
of "individuals" across time or space. 
This interpretation is probably one reason for its popularity in ecology. 
\citet{cohen2016population} discuss the interpretation of the exponent (also called slope parameter) of Taylor's power law.
In scientific use and testing of TL, \eqref{main} is the most basic relation against many plausible alternative models.
There are many extensions of TL such as the quadratic TL of \citet{taylor1978density} which can be more appropriate than \eqref{main}, as shown, for instance, in \citet{xu2021spatial}
for US county populations. Sometimes a cubic generalization of TL is required to describe data, as in \citet{yang2022temporal} for the study of mortality rates in Japan. 
And sometimes the variance function of the negative binomial distribution $variance = a\cdot mean+ b\cdot mean^2$, 
as first proposed by \citet{bartlett1936some}, is superior to any form of TL, as shown in \citet{cohen2017linking}. 
Fitting TL \eqref{main} is a step that follows exploratory data analysis of multiple alternative models, each of which poses its own problems of estimating parameters.

Whatever the form of TL, estimation of unknown parameters roughly follows the same idea. 
We compute sample means and sample variances along one dimension (say, space)
of a rectangular array of nonnegative data indexed by (say) space and time.
This produces a pair (sample mean, sample variance) for each point in time.
The sample mean and sample variance are assumed to be positive so their logarithms are well defined.
The slope and intercept of log variance as a linear function of log mean are calculated,
typically by ordinary least squares but sometimes by more robust methods 
\citep{yang2022temporal}
to optimize the fit of Taylor's power law to these pairs of sample moments. 

For example, in ecological species abundance data, the spatial TL, on which we will focus in the present paper, 
assumes that the time-varying mean and variance of the marginal distribution are linked at each time $t$ by \eqref{main}
and these population quantities are approximated by the corresponding sample moments computed from abundance data of a given species observed at different locations for each time $t$. 
After taking the logarithm of \eqref{main}, this leads to the least-squares estimator given in Section \ref{sec2}. 

A simple log-log plot (log variance on the vertical axis, log mean on the horizontal axis) is the classical graphical tool used by practitioners to assess whether TL describes the data. 
We shows here that the confidence intervals of parameter estimates
resulting from ordinary least squares may be seriously biased, 
and we show how to correct that bias.
Section \ref{sec6} illustrates with real data.

A wide statistical literature links power variance laws like TL with 
some classes of distributions  \citep{Tw1984, BL1986,BLS1987}; those classes are compared in \cite{BL2020}. \cite{ch2022} provide infinitely divisible examples of the Bar-Lev Enis type.
However, formal statistical analyses of TL are rather scarce.  
Some methodological contributions include \citet{clark1994small},
\citet{james2007fitting}, \cite{jorgensen2011bias}, 
\citet{perry1992fitting} \citet{sakoda2019tracking}, and \citet{shi2017comparison}, 
among others. 
Beyond the previous papers which state that TL holds for interesting classes of distributions, these
 contributions mainly focus on inference procedures or computational issues and 
compare different approaches for estimating TL parameters in various situations such as sparsity, small samples, or data with covariates.
None of these contributions discussed formally the large-sample behavior of 
least-squares or likelihood estimators, which
is precisely the scope of the present paper.

We believe that the statistical analysis made here can guide the analysis of more complex versions of TL. 
Our initial intention was to study some detailed conditions that justify the commonly used method of log variance-log mean least-squares regression. 
However, our main assumptions and results could be easily extended. 
Our approach to studying TL requires a careful analysis of a specific minimum of contrast estimator derived from some sample moments instead of the usual pair response/covariate used in standard regression analysis. 
Analyzing other methods such as likelihood estimation or more general models that use log-regression could be conducted using similar arguments.

This work can be useful for both quantitative ecologists and statisticians. 
Our analysis can help ecologists to measure uncertainty in the estimated parameters with precise confidence bands. It will prevent them from using standard 
but biased methods to compute confidence bands for TL parameters. 
{
By biased, we mean a non-centered asymptotic distribution for TL parameters}.
In Section \ref{sec5}, it is shown that in some cases, the Gaussian asymptotic distribution of least-squares estimates is not centered. 
For statisticians, our work reveals that a popular family of stochastic relations widely studied in several sciences is 
a source of interesting new theoretical problems, in particular for getting consistency properties or for constructing testing problems in a nonstationary spatiotemporal framework. 

Many extensions of our results are possible, for example to
the numerous variants of TL discussed above.
In addition, this paper will focus only on distributions with light tails, i.e. when the probability distributions of interest have some finite moments of order $k>2$. 
Though sufficient to prove rigorous consistency results for our estimates, such a restriction is not necessary for TL to occur. 
TL may hold for the sample mean and sample variance from stable laws
with infinite variance or infinite mean, as shown in \citet{brown2017taylor}, \citet{brown2021taylor}, \citet{cohen2020heavy} or \citet{cohen2022covid}.
However, even with our restriction to light-tailed distributions, 
ours is the first rigorous analysis of the asymptotic properties of the least-squares estimators of Taylor's power law parameters.

The paper is organized as follows. Section \ref{sec2} specifies our model and the least-squares estimators of TL parameters.
Some numerical experiments in Section \ref{sec3} assess the quality of these estimates, depending on various assumptions made 
about the process that generates the data.
Section \ref{sec4} gives consistency properties and Section \ref{sec5} gives limiting distributions for our estimators.
Section \ref{sec6} illustrates the application of our results to real data.
Section \ref{sec7} summarizes our conclusions. 
All proofs of our results are postponed to Section \ref{sec8}.
Section \ref{sec9} contains some technical lemmas used in our proofs.

\section{Least-squares estimators for Taylor's law parameters}\label{sec2}

Let $(X_{j,t})_{1\leq j\leq n,1\leq t\leq T}$ be an array of nonnegative random variables { indexed by integers $j,\ t$}. 
These variables can take integer values, for instance,
when $X_{j,t}$ represents the number of individuals 
observed at time $t$ and at the location $j$, 
or more generally nonnegative real values, for instance, when $X_{j,t}$ represents
the number of individuals observed at time $t$ divided by the area of site $j$,
in which case $X_{j,t}$ represents a number of individuals per unit of area (population density). 

We assume that the spatial variance $\sigma_t^2:=\v(X_{j,t})$ 
(that is, the variance of the $n$ elements in column $t$ of the array $(X_{j,t})$)
and the spatial mean $\mu_t:=\E\left(X_{j,t}\right)$ 
(that is, the mean of the $n$ elements in column $t$ of the array $(X_{j,t})$)
at time $t$ do not depend on the location $j$,
are strictly positive, and satisfy
\begin{equation}\label{main}
\sigma_t^2=\alpha \mu_t^{\beta},\quad 1\leq t\leq T,
\end{equation} 
for a pair of parameters $(\alpha,\beta)\in (0,\infty)\times \R$.
In most ecological applications, the estimate of $\beta$ is positive. 
A value $\beta>1$ is sometimes interpreted as indicating aggregation or clumping of the organisms 
because a family of Poisson distributions in which the mean varies in time and remains constant over space obeys $\alpha=\beta=1$.
Values below $1$ are uncommon but possible,
for example, in models of a population in a Markovian environment, e.g., \citet{Cohen2014a, Cohen2014b}
and numerical simulations of some fisheries models gave $\beta<0$ \citep{Masami2015}.

The population parameters $(\mu_t,\sigma_t^2)$ are estimated by their empirical counterparts from the sample $X_{1,t},\ldots,X_{n,t}$. See below for details.
This form of TL is known as a spatial TL because the empirical mean and the empirical variance depend on $n$ spatial locations at each time.
We do not consider here 
other variants of TL, such as the temporal TL,
in which the mean and the variance depend on $T$ time points at each spatial location, or the
hierarchical spatial TL or the hierarchical temporal TL. 
For instance, the temporal TL assumes that the population moments
$\sigma_j^2=\v(X_{j,t})$ and $\mu_j=\E X_{j,t}$ depend only on $j$ and not on $t$. 
The corresponding power law is
$$\sigma_j^2=\alpha \mu_j^{\beta},\quad 1\leq j\leq n.$$
Asymptotic properties of the temporal TL can be obtained with similar arguments. 
This is why we mainly focus on the spatial law \eqref{main}.

The main problem here is inference about the two parameters $\alpha$ and $\beta$ of \eqref{main}. 
A natural way to proceed is to use a linear regression of the logarithm of 
the empirical spatial variance computed at each time $t$ as a function of
the logarithm of the empirical spatial mean computed at each time $t$.
Let 
$$\hat{\mu}_t=\frac{1}{n}\sum_{j=1}^n X_{j,t},\quad \hat{\sigma}_t^2=\frac{1}{n}\sum_{j=1}^n X_{j,t}^2-\hat{\mu}_t^2\,.$$ 
{While the variance can be also estimated from the unbiased estimator
$$\hat{\sigma}_t^2=\frac{1}{n-1}\sum_{j=1}^n (X_{j,t}-\hat{\mu}_t)^2\, ,$$ 
our asymptotic results will not depend on the choice of a biased or unbiased version of the sample variance.
}
Since \eqref{main} implies
\begin{equation}
\log\sigma_t^2=\log(\alpha)+\beta \log\left(\mu_t\right), \label{loglogTL}
\end{equation}
it is natural to estimate the pair $\theta=\left(\log(\alpha),\beta\right)'$
 from the ordinary least-squares estimator 
$$\hat{\theta}_n=\underset{\theta\in\R^2}{\arg\min} \sum_{t=1}^T\left(\log\hat{\sigma}_t^2-\theta_1-\theta_2\log\hat{\mu}_t\right)^2.$$
Explicitly,
\begin{eqnarray}\label{defest}\hat{\theta}_n&=&\hat{D}_T^{-1}\hat{N}_T,\qquad\mbox{ where} \\
\label{defD}
\hat{D}_T&=&\begin{pmatrix}\displaystyle 1&\frac{1}{T}\sum_{t=1}^T 
\log\hat{\mu}_t\\ \frac{1}{T}\sum_{t=1}^T \log\hat{\mu}_t\displaystyle& 
\frac{1}{T}\sum_{t=1}^T \log^2 \hat{\mu}_t\end{pmatrix},
\quad
\hat{N}_T=\begin{pmatrix} \frac{1}{T}\sum_{t=1}^T \log\hat{\sigma}^2_t\\ \frac{1}{T}\sum_{t=1}^T\log\hat{\mu}_t\cdot\log\hat{\sigma}_t^2\end{pmatrix}.
\end{eqnarray}
In the above formula, $\log \hat{\mu}_t$ or $\log \hat{\sigma}_t$ should be replaced by $0$ if 
$\hat{\mu}_t=0$ or $\hat{\sigma}^2_t=0$, respectively. 
For simplicity, we do not introduce an indicator function in the above formula.
If we set
$$D_T:=\begin{pmatrix} 1&\frac{1}{T}\sum_{t=1}^T \log\mu_t\\\frac{1}{T}\sum_{t=1}^T \log\mu_t& \frac{1}{T}\sum_{t=1}^T\log^2 \mu_t\end{pmatrix},\quad N_T:=\begin{pmatrix} \frac{1}{T}\sum_{t=1}^T \log\sigma^2_t\\ \frac{1}{T}\sum_{t=1}^T\log\mu_t\cdot\log\sigma_t^2\end{pmatrix},$$
then analogously $\theta=D_T^{-1}N_T$, which makes it natural to introduce the above estimator.
{
Throughout this paper, we will assume that the time series length $T$ depends on the number $n$ of spatial sites}, 
i.e. $T:=T_n$. 
We will study the asymptotic properties of $\hat{\theta}:= \hat{\theta}_n$ in \eqref{defest} as $n\to\infty$.

\section{Numerical experiments}\label{sec3}

Before stating our main assumptions and results, we investigate the accuracy of the least-squares estimates discussed in the previous section through some numerical experiments. 
We will consider several scenarios for which the moment conditions on the marginal distribution are not the same.
We also study a scenario with some zero-inflated data. 
Throughout this section, we consider a very simple time-varying parameter $\lambda_t=\exp\left(\cos(t)\right)$, $1\leq t\leq T$, to define our different scenarios. 
Moreover, for simplicity, we will assume that 
all the ``observations'' $\{X_{j,t}\}$ are independent across $j$ and $t$, i.e. there is no spatial or temporal dependence for the different random variables. 
To get a meaningful approximation of the population mean
and population variance from empirical (or sample) counterparts, we assume 
that for any time point $t$, the random variables $X_{1,t},\ldots,X_{n,t}$ are { identically} distributed, i.e., all have the same distribution.
We use DGP to abbreviate ``data-generating process''.
{
We use the relation $\sim$ always to mean that the random variables on both sides have the same probability distribution (\footnote{Here we will never  mean that deterministic functions on both sides are asymptotically equivalent as some quantity approaches a limit.}).}

\begin{enumerate}
\item
{\bf DGP1}. Poisson $X_{j,t} \sim \mathcal{P}(\lambda_{t})$. In this case, $\sigma_t^2=\mu_t$ and $(\alpha,\beta)=(1,1)$. 
All the exponential moments are finite here, i.e. $\E\exp\left(L X_{j,t}\right)<\infty$ for any $L>0$.
\bigskip

\item
{\bf DGP2}. Squares of Gaussian $X_{j,t} \sim \mathcal{N}(0,\lambda_{t})^2$ (also known as $\chi^2$ with one degree of freedom). In this case, $\sigma_t^2= 2 \mu_t^2$ and $(\alpha,\beta)=(2,2)$. 
Some exponential moments are infinite: $\E\exp\left(L X_{j,t}\right)<\infty$ if and only if $L\lambda_t<1/2$.
\bigskip

\item
{\bf DGP3}. Mixture of Poisson $X_{j,t} \sim \mathcal{P}(\lambda_{t}\times Z_{j,t}+\zeta_{t})$, with $Z_{j,t}\sim \sqrt{3/4}\exp\left(\mathcal{N}(0,1)^2/8\right)$ and setting $v=\v\left(Z_{j,t}\right)$,
$$\zeta_{t}=\frac{1+\sqrt{1+4 v^2\lambda_{t}^2}}{2v}-\lambda_{t}.$$
In this case, $(\alpha,\beta)=(v,2)$ and
the fourth moment is infinite: 
$\E X_{j,t}^L<\infty$ if and only if $0\leq L<4$.
\end{enumerate}

\begin{figure}[H]
     \centering
     \begin{subfigure}[b]{0.5\textwidth}
         \begin{center}
         \includegraphics[width=8cm,height=5cm]{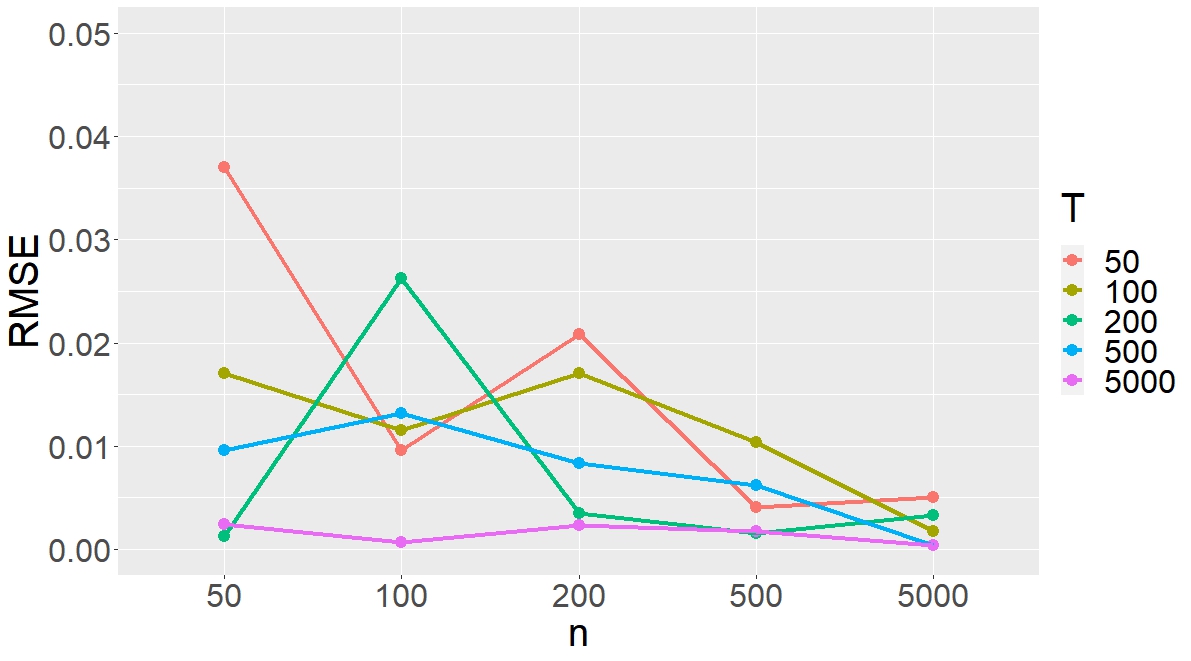}
         \caption{{\bf DGP1}}
				\end{center}
     \end{subfigure}
     \bigskip
		
     \begin{subfigure}[b]{0.5\textwidth}
         \centering
         \includegraphics[width=8cm,height=5cm]{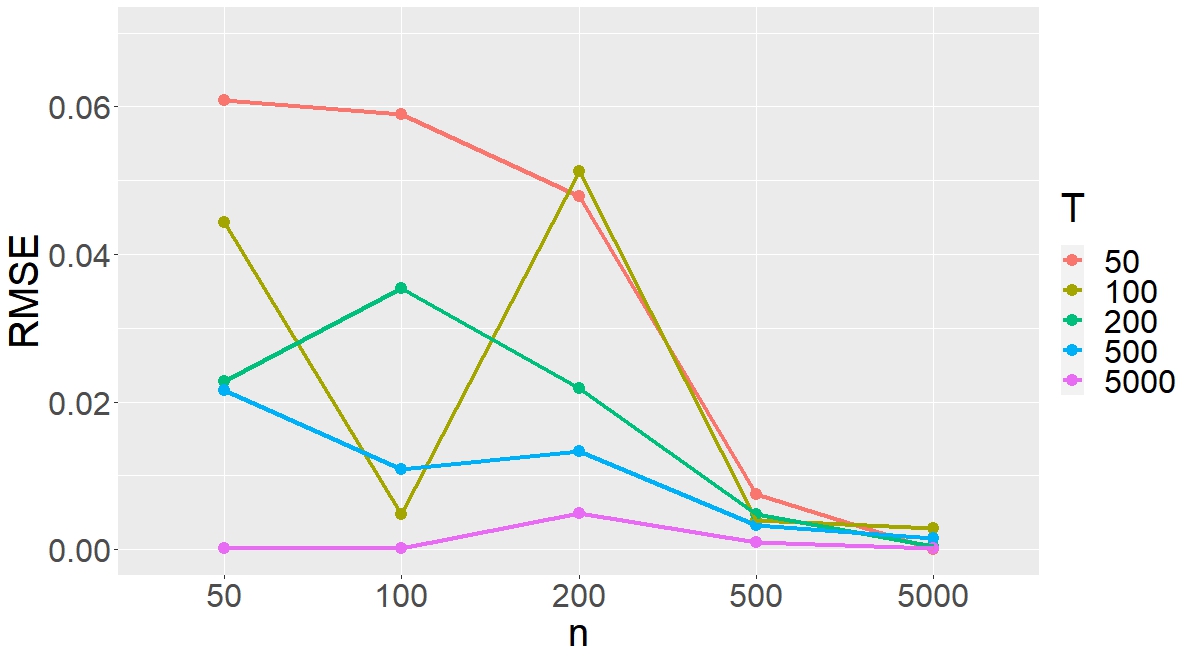}
         \caption{{\bf DGP2}}
     \end{subfigure}
     \bigskip
		
     \begin{subfigure}[b]{0.5\textwidth}
         
         \includegraphics[width=8cm,height=5cm]{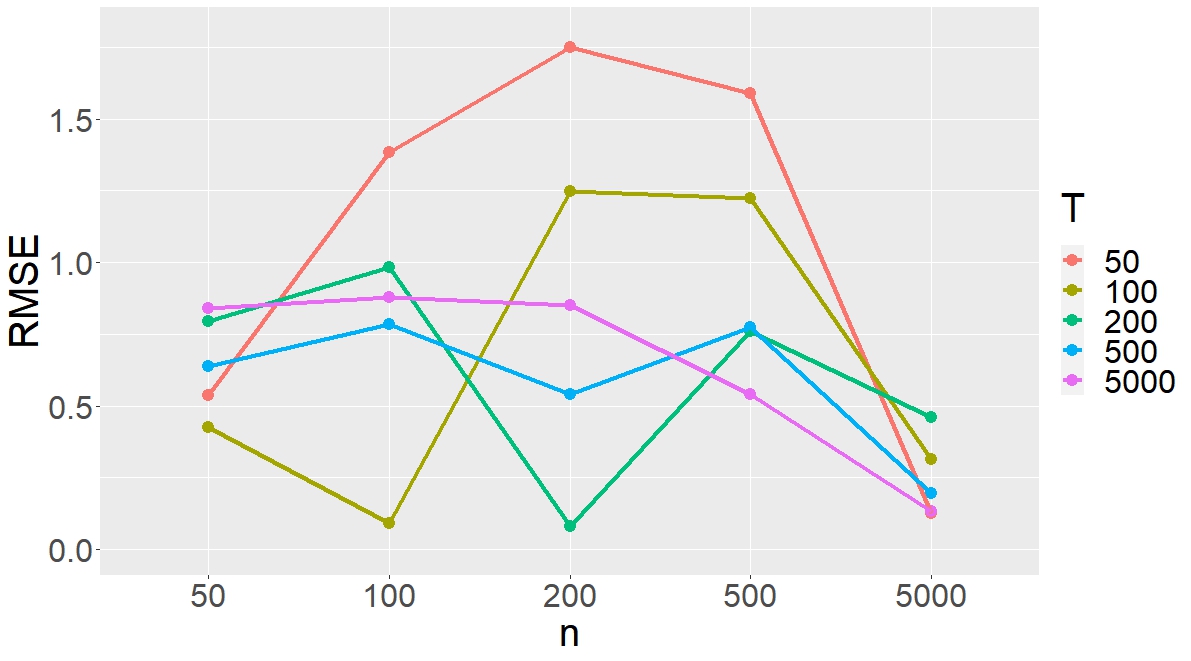}
         \caption{{\bf DGP3}}
     \end{subfigure}
        \caption{Average RMSE values of the spatial TL for different values of $n$ and $T$}
        \label{RMSE_simu1}
\end{figure}

{Figure \ref{RMSE_simu1} represents the root mean squared error 
(RMSE) of $\hat{\beta}$ for the three scenarios and different values of the pair $(n,T)$. 
Large values of $n$ and $T$ seem to lead to a more accurate estimation of the slope parameter. 
However, the tail of the marginal distribution seems to affect this accuracy: 
the third scenario with infinite fourth moments leads to the largest estimation errors and large values of $T$ can 
reduce the accuracy if the value of $n$ is moderate. 
{
The difference in the accuracy of the three DGPs can be seen from the differences in the vertical scales in the three panels of Figure \ref{RMSE_simu1}}.
This observation was expected since, in theory, 
the moment conditions affect how closely empirical moments 
approximate population moments.
It is not surprising to get smaller RMSE when the deviations $\hat{\mu}_t-\mu_t$ and $\hat{\sigma}_t^2-\sigma_t^2$ can be made small.

We next evaluate how the number of $0$ counts in the sample affects RMSE. 
 Counts equal to $0$ are very common in abundance data.
For example, 
Figure \ref{serie_temp} 
{shows a} real data set of fisheries abundances obtained from the 
ICES International Bottom Trawl Survey for the North Sea (NS-IBTS), 
which will be described  in more detail in Section \ref{sec5}.

We now construct a new DGP with various proportions of $0$ counts
to mimic zero-inflated data. 
Let $W_{j,t}, 1\leq j\leq n$, $1\leq t\leq T$, be Bernoulli random variables 
independently and identically distributed (iid) as
$\mathcal{B}(p), p\in (0,1)$,
where $p :=\Pr\{W_{j,t}=1\}$. 
Thus $0$ values appear with probability $1-p$. 
Set $X_{j,t}=W_{j,t}X'_{j,t}$ where $X_{j,t}'$ are defined in {\bf DGP2}
and where $W_{j,t}$ are independent of $X'_{j,t}$.
It may be shown that the spatial TL is still satisfied with $\beta=2$.
However, $\alpha = 3p - p^2$, which reduces to $\alpha = 2$ if $p=1$.

Figure \ref{thinning} shows that $p=0.2$ gives 
less accurate results  than $p=0.5$. 
This observation is not surprising because 
smaller $p$ leads to more $0$ values, and
{smaller values of the empirical means lead to larger absolute values for their logarithm 
and more variability is expected for the partial sums defining the log-regression estimates.}

\begin{figure}[H]
\begin{center}
        \includegraphics[scale=0.65 ]{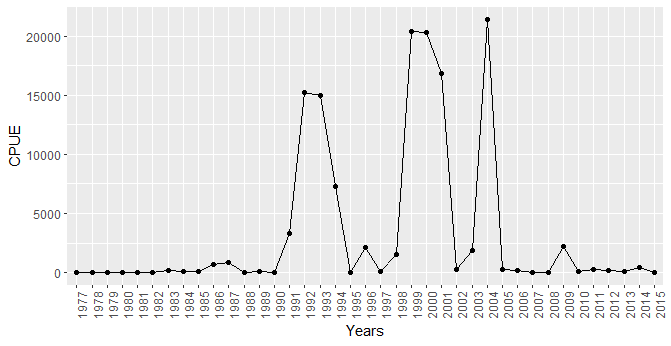}
        \caption{Time series for the species \textit{Clupea harengus} on the subarea 31F2. {CPUE is counts per unit effort.} }
        \label{serie_temp}
\end{center}
\end{figure}

\begin{figure}[H]
     \centering
     \begin{subfigure}[b]{0.5\textwidth}
         \centering
         \includegraphics[width=8cm,height=5cm]{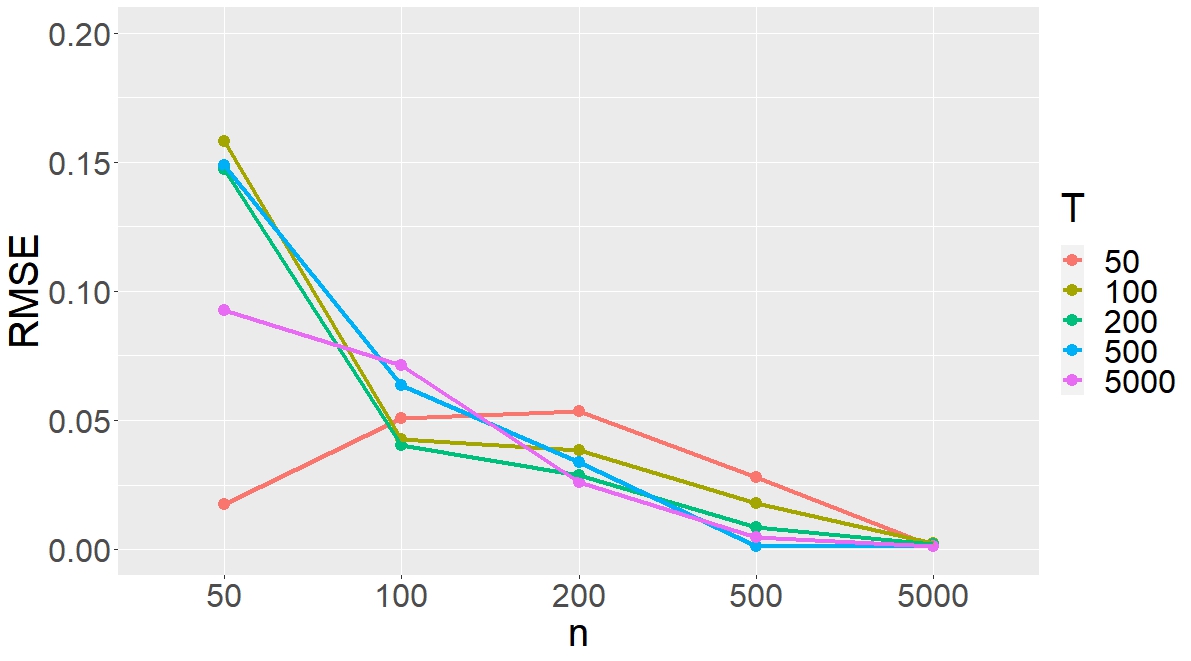}
         \caption{$X_{j,t}\sim \mathcal{N}\left(0,\lambda_t\right)^2\times \mathcal{B}(0.2)$}
     \end{subfigure}
    \bigskip		
		
     \begin{subfigure}[b]{0.5\textwidth}
         \centering
         \includegraphics[width=8cm,height=5cm]{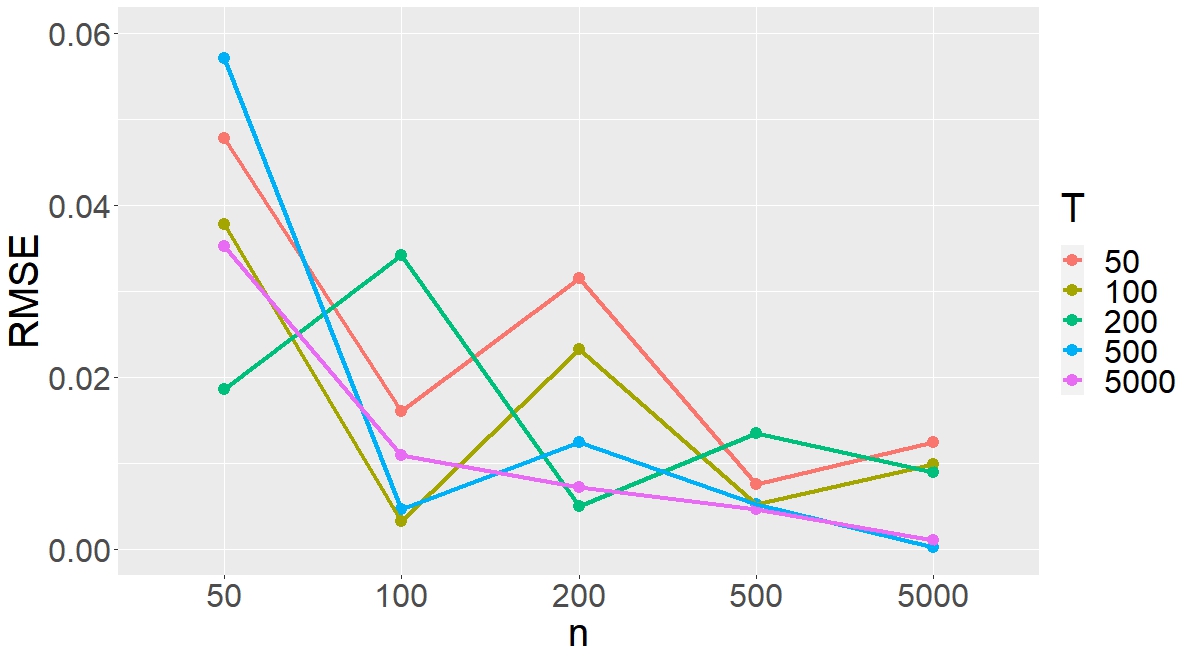}
         \caption{$X_{j,t} \sim \mathcal{N}\left(0,\lambda_t\right)^2\times\mathcal{B}(0.5)$}
     \end{subfigure}
     
        \caption{Effect of zero-inflation on the RMSE of $\hat{\beta}$ for DGP2}
				\label{thinning}
        \end{figure}

\section{Consistency properties}\label{sec4}

We now study the consistency of least-squares estimates. 
Consistency is defined to mean that $\lim_{n\rightarrow\infty}\hat{\theta}=\theta$ in probability 
(weak consistency) or almost surely (strong consistency).  

We assume that the time horizon $T$ of the time series depends on $n$, i.e. $T=T_n$.
If $T_{+}:=\sup_{n\geq 1} T_n<\infty$, consistency properties will roughly follow from the law of large numbers. 
It is more challenging to study the case $T_n\rightarrow \infty$, which is also useful in applications because the dimensions of the arrays of data are often balanced.
Though time series with length converging to infinity will be the main focus, our results are also valid for finite $T_{+}$.

\subsection{General results for nonnegative random variables}
To derive the asymptotic properties of our estimates, we assume:

\begin{description}
\item[A1] The time series $\left(X_{j,t}\right)_{1\leq t\leq T}$, $1\leq j\leq n$, are independent and identically distributed.
\item[A2] There exists $k>2$ such that $\sup_{t\geq 1}\E\vert X_{1,t}\vert^k<\infty$.
\item[A3] There exists some $\delta>0$ such that 
$$p(\delta):=\inf_{t\geq 1}\P\left(X_{1,t}>\delta\right)>0  \quad \mathrm{and}\quad 
q(\delta)=\inf_{t\geq 1}\P\left(\left\vert X_{1,t}-X_{2,t}\right\vert> \delta\right)>0.$$
\item[A4] There exists $c>0$ such that for any positive integer $T$, $$\frac{1}{T}\sum_{t=1}^T \left(\log \mu_t\right)^2-\left(\frac{1}{T}\sum_{t=1}^T\log \mu_t\right)^2\geq c.$$
\end{description}

Let us comment on our assumptions.
\begin{enumerate}
\item
Assumption {\bf A1} does not require any specific dependence property across time. 
Independence across sites $j$ is needed only for technical reasons. 
However, spatial correlations can be observed in data, as for instance in a study of Hokkaido voles \citep{cohen2016population}.  
In theory, it may be possible 
to remove the assumption of independence across sites $j$ 
by assuming a stationarity condition and some weak dependence property between the different sites $j$. 
However, getting consistency properties under such conditions would require additional assumptions,
in particular to get the law of large numbers when the locations $j$ are irregularly spaced. 
We prefer to keep the setup simple. 

\item
The technical moment assumption {\bf A2} is needed to prove limit results as in Theorem \ref{th1}.
\item
Assumption {\bf A3} roughly means that Dirac masses cannot be cluster points of the marginal distribution, 
since {\bf A3} is not satisfied if and only if 
at least one of the two sequences $\left(X_{1,t}\right)_{t\geq 1}$ or $\left(X_{1,t}-X_{2,t}\right)_{t\geq 1}$ has a subsequence converging to $0$ in probability. 
Excluding Dirac masses is important to avoid small values for the population mean and the population variance 
as well as for their empirical versions and then large negative values for their logarithm.
\item  
Assumption {\bf A4} ensures that the vectors $\left(\log \mu_t\right)_{1\leq t\leq T}$ vary sufficiently in time. 
Such a condition is needed 
to avoid some pathological behaviors such as the convergence of $\mu_t$ when $t\rightarrow \infty$. Without this condition, the matrix $\hat{D}_T$ can be asymptotically singular. For a standard linear regression model $y_t=x_t'\beta+\varepsilon_t$, our condition is comparable to the existence of a lower bound for the determinant of the matrices $T^{-1}\sum_{t=1}^Tx_tx_t'$ classical in statistical learning. 
\end{enumerate}


\begin{theo}\label{th1}
Given Assumptions {\bf A1-A2-A3-A4}, 
there exists $\gamma>0$ such that if $T_n=O\left(\exp\left(\gamma n\right)\right)$, then 
$\lim_{n\rightarrow \infty}\hat{\theta}=\theta$ in probability. 
If Assumption {\bf A2} holds with $k>4$, 
then $\lim_{n\rightarrow \infty}\hat{\theta}=\theta$ almost surely (a.s.).
\end{theo}

Theorem \ref{th1} requires a limited growth of the time series length $T$ with respect to the number $n$ of sites. 
Another assumption concerning the marginal distribution of $X_{1,t}$ avoids such a restriction and leads to consistency. 
Define $\log_{-}(x):=-\log\min(x,1)$ for any $x>0$ and $\log_{-}^q(x):=(\log_{-}(x))^q$ for any $q>0$.

\begin{description}
\item[A5] There exists $q>2$ such that 
$$\sup_{t\geq 1}\E\mathds{1}_{X_{1,t}>0}\cdot\log^q_{-} X_{1,t}<\infty,\quad  
\sup_{t\geq 1}\E\mathds{1}_{X_{1,t}\neq X_{2,t}}\cdot\log^q_{-}\vert X_{1,t}-X_{2,t}\vert<\infty.$$
\end{description}

\begin{theo}\label{th2}
Given Assumptions {\bf A1-A2-A3-A4-A5}, $\hat{\theta}_n$ is weakly consistent if $k>2$ and strongly consistent if $k>4$ in {\bf A2}.
\end{theo}

\subsection{A more specific result for integer-valued random variables}

For integer-valued random variables, the previous assumptions can be simplified. 
In particular, it is straightforward to show that Assumption {\bf A5}
is automatically satisfied. 
Assumption {\bf A3} becomes:

\begin{description}
\item[A3$'$ ]
We have $\inf_{t\geq 1}\P\left(X_{1,t}>0\right)>0$ and $\inf_{t\geq 1}\P\left(X_{1,t}\neq X_{2,t}\right)>0$. 
\end{description}

Then:

\begin{cor}
Given Assumptions {\bf A1-A2-A3$'$-A4}, $\hat{\theta}_n$ is weakly consistent if $k>2$ and strongly consistent if $k>4$ in {\bf A2}.
\end{cor}

\section{Asymptotic normality of parameter estimates}\label{sec5}
{
We now derive asymptotic normality of our log-regression OLS estimators. To keep the setup as simple as possible, we will also assume independence in the time direction, i.e. for any $j\geq 1$, the $X_{j,1},X_{j,2},\ldots$ are assumed to be mutually independent. 
This assumption can be relaxed but 
this simple setup suffices to give insight into the
weak convergence of such estimates.}

Define $Y_{j,t}:=(X_{j,t},X_{j,t}^2)'$ 
and $m_t:=\E Y_{1,t}$ for $1\leq j\leq n$ and $1\leq t\leq T$.
Define the empirical version of $m_t$ by 
$$\hat{m}_t:=\frac{1}{n}\sum_{j=1}^nY_{j,t}.$$
The two components of $m_t$ and $\hat{m}_t$ will be denoted respectively by $m_{1,t},m_{2,t}$ and $\hat{m}_{1,t},\hat{m}_{2,t}$. 
In terms of 
our previous notations, we have $m_{1,t}=\mu_t$ and $\hat{m}_{1,t}=\hat{\mu}_t$.
In what follows, we will use both notations, depending on the context.

We also use a standard notation for weak convergence. When $(Z_n)_{n\geq 1}$ is a sequence of random vectors of $\R^2$ and $\nu$ is a probability measure on the Borel $\sigma$-field of $\R^2$, notation $Z_n\Rightarrow \nu$ means that $Z_n$ converges weakly to $\nu$.

In the following analysis, we are thinking of a function $\phi_{\theta}$ (with parameter $\theta\in \R^2$) from domain $(0,\infty)^2$ into range $\R^2$.
In the 2-dimensional argument $(x,y)$ of $\phi_{\theta}$, $x$ represents the empirical 
or expected spatial mean of the observations $X_{j,t}$ and $y$ represents 
the empirical or expected spatial mean of the squared observations $X_{j,t}^2$.
With either empirical or expected mean, $y-x^2$ represents the variance, i.e., the mean squared observation minus the squared mean observation.
The two parameters $\theta_1=\log(\alpha),\ \theta_2=\beta$ of $\phi_{\theta}$ correspond to the intercept and the slope of the log-log form
\eqref{loglogTL} of TL.
The first component $\phi_{1,\theta}(x,y):=\log(y-x^2)-\theta_1-\theta_2\log(x)$ of $\phi_{\theta}$ is the deviation from \eqref{loglogTL},
conditional on the two parameter values. 
It corresponds to the first component of the normal equations of ordinary least squares estimation, \eqref{defest} and
\eqref{defD}.
The second component $\phi_{2,\theta}(x,y):=\log(x)\log(y-x^2)-\theta_1\log(x)-\theta_2(\log(x))^2$
similarly corresponds to the second component of the normal equations of ordinary least squares.

Setting $\theta=\left(\log(\alpha),\beta\right)'$, the estimation error for the log-regression parameters has the following decomposition:

\begin{eqnarray*}
\hat{\theta}-\theta&=&\hat{D}_T^{-1}\hat{N}_T-D_T^{-1}N_T\\
&=& \hat{D}_T^{-1}\left(\hat{N}_T-N_T+\left(D_T-\hat{D}_T\right)\theta\right).
\end{eqnarray*}
This decomposition implies that 
\begin{equation}\label{UT}
U_T:=\hat{D}_T(\hat{\theta}-\theta)=\frac{1}{T}\sum_{t=1}^T\left(\phi_{\theta}(\hat{m}_t)-\phi_{\theta}(m_t)\right),
\end{equation}
where
\begin{eqnarray*}
\phi_{\theta}(x,y)&:=&\begin{pmatrix} \phi_{1,\theta}(x,y)\\\phi_{2,\theta}(x,y)\end{pmatrix},\quad x>0,y>0.\\
\end{eqnarray*}
{
The idea is simply to make a Taylor expansion of the function $\phi_{\theta}$ with respect to its two arguments representing the spatial means of the two samples $(X_{j,t})_{1\leq j\leq n}$ and $(X^2_{j,t})_{1\leq j\leq n}$}
respectively, computed at any time point $1\leq t\leq T$. We will then approximate $U_T$ with its first-order Taylor expansion
\begin{equation}\label{VT}
V_T=\frac{1}{T}\sum_{t=1}^T J_{\phi_{\theta}}(m_t)\cdot \left(\hat{m}_t-m_t\right),
\end{equation}
where $J_{\phi_{\theta}}(x,y)$ denotes the Jacobian matrix of $\phi_{\theta}$ at point $(x,y)$.
Intuitively, the convergence rate for $V_T$ is $\sqrt{nT}$ since $V_T$ contains a double summation with $n\times T$ centered terms.
We then define 
\begin{equation}\label{stablevar}
\Gamma_T:=\v(\sqrt{nT} V_T)=\frac{1}{T}\sum_{t=1}^T J_{\phi_{\theta}}(m_t)\cdot \v\left(Y_{1,t}\right)\cdot J_{\phi_{\theta}}(m_t)'
\end{equation}
and $C_T:=\Gamma_T^{-1/2}$. It is expected that $\sqrt{nT}C_TV_T\Rightarrow \mathcal{N}_2(0,I_2)$.
The remaining step is to investigate the asymptotic behavior of $\sqrt{nT}\left(U_T-V_T\right)$. 
When $T/n\rightarrow 0$, one can show that under the assumptions of Theorem \ref{th1}, $\sqrt{nT}\left(U_T-V_T\right)=o_{\P}(1)$. 
However, when $T/n\rightarrow s>0$,
this is not the case, but under the assumptions of Theorem \ref{th1}, one can show that
$$\sqrt{nT}\left(U_T-V_T\right)=\sqrt{\frac{n}{T}}\sum_{t=1}^T\begin{pmatrix}(\hat{m}_t-m_t)'\nabla^{(2)}\phi_{1,\theta}(m_t)(\hat{m}_t-m_t)\\(\hat{m}_t-m_t)'\nabla^{(2)}\phi_{2,\theta}(m_t)(\hat{m}_t-m_t)\end{pmatrix}+o_{\P}(1),$$
where for a twice-differentiable function $f:\R^2\rightarrow \R$, we denote by $\nabla^{(2)}f(m)$ the Hessian matrix of $f$ at point $m$.
The first right-hand term in the previous equality 
does not vanish asymptotically and has an expectation of order $\sqrt{T/n}$. Hence this non-negligible bias has to be included for deriving the asymptotic distribution of $\hat{\theta}$.

\begin{theo}\label{NA}
Let $\lambda_{-}(\Gamma_T)$ denote the smallest eigenvalue of $\Gamma_T$ defined in \eqref{stablevar}.
Suppose that Assumptions {\bf A1-A2-A3-A4} hold with $k>4$ in {\bf A2} {
and for any integer $j\geq 1$, $(X_{j,t})_{t\geq 1}$ is a sequence of independent random variables. Assume furthermore} that $\varliminf_n \lambda_{-}(\Gamma_T)>0$. 
\begin{enumerate}
\item
If $T_n/n\rightarrow 0$, then 
$$\sqrt{nT}C_T\hat{D}_T\left(\hat{\theta}-\theta\right)\Rightarrow \mathcal{N}_2\left(0,I_2\right).$$
\item
If $T_n/n=O(1)$, then
$$\sqrt{nT}C_T\hat{D}_T\left(\hat{\theta}-\theta\right)-\frac{1}{2}\sqrt{\frac{T}{n}}C_T E_T\Rightarrow \mathcal{N}_2\left(0,I_2\right),$$
where $E_T:=\frac{1}{T}\sum_{t=1}^TG_t$ and
$$G_t:=\left(\mbox{ Tr }(\nabla^{(2)}\phi_{1,\theta}(m_t)\v(Y_{1,t})),\mbox{ Tr }(\nabla^{(2)}\phi_{2,\theta}(m_t)\v(Y_{1,t}))\right)'.$$
\end{enumerate}
\end{theo}

{
We comment on these results.}
\begin{enumerate}
\item
The existence of a lower bound for the eigenvalues of $\Gamma_T$ is always difficult to check 
for a specific sequence of marginal distributions, though quite natural. 
We provide sufficient conditions as well as a way to check it for Poisson random variables in Lemma \ref{defpos} in the Appendix.
\item
When $T_n/n\to k>0$ as $n\rightarrow \infty$,
the normalized statistics 
$\sqrt{nT}C_T\hat{D}_T(\hat{\theta}-\theta)$ in Theorem \ref{NA}.1 become asymptotically biased. 
Neither the matrix $C_T$ nor the vector 
$E_T$ converges in general when $n\rightarrow \infty$ 
but these sequences 
remain bounded under our assumptions. This non-centering of the asymptotic distribution is atypical 
for regular statistical models and is due here to the fact that when the array of data remains balanced, 
there are not enough replicates of the time series to make the second-order term in the Taylor expansion of the estimates negligible.
The asymptotics when $T/n=O(1)$ (Theorem \ref{NA}.2) also hold when $T/n=o(1)$, since 
then $\sqrt{T/n}C_T E_T=o(1)$ becomes negligible.
\item  {
Figure \ref{biasst}  shows the effect of neglecting the bias for simulated Poisson data when the number of sites equals  the number of time points. 
Without bias correction, the probability distribution of the two coordinates of the normalized statistics $\sqrt{nT}\hat{C}_T\hat{D}_T(\hat{\theta}-\theta)$ are not well approximated by that of a standard Gaussian.}

\item
To use these results in practice for constructing tests or confidence bands, it is necessary to estimate 
$\Gamma_T$ and $E_T$ from the data.
This can be done with no additional assumptions: 
\end{enumerate}

\begin{cor}\label{CB}

Suppose that all the assumptions of Theorem \ref{NA} hold true with either the asymptotics $T/n=o(1)$ or $T/n=O(1)$. 
Set $\Sigma_t:=\v\left(Y_{1,t}\right)$ and 
$$\hat{\Sigma}_t:=\frac{1}{n}\sum_{j=1}^n Y_{j,t}Y_{j,t}'-\overline{Y}_{n,t}\overline{Y}_{n,t}',\quad \overline{Y}_{n,t}:=\frac{1}{n}\sum_{j=1}^nY_{j,t},$$
$$\hat{\Gamma}_T:=\frac{1}{T}\sum_{t=1}^T J_{\phi_{\hat{\theta}}}(\hat{m}_t)\hat{\Sigma}_tJ_{\phi_{\hat{\theta}}}(\hat{m}_t)',\quad \hat{C}_T:=\hat{\Gamma}_T^{-1/2},$$
$$\hat{E}_T:=\frac{1}{T}\sum_{t=1}^T \hat{G}_t,$$
$$\hat{G}_t:=\left(\mbox{Tr}(\nabla^{(2)}\phi_{1,\hat{\theta}}(\hat{m}_t)\hat{\Sigma}_t),\mbox{Tr}(\nabla^{(2)}\phi_{2,\hat{\theta}}(\hat{m}_t)\hat{\Sigma}_t)\right)'.$$
Then
\begin{equation}\label{estimee}
\sqrt{nT}\hat{C}_T\hat{D}_T\left(\hat{\theta}-\theta\right)-\frac{1}{2}\sqrt{\frac{T}{n}}\hat{C}_T \hat{E}_T\Rightarrow \mathcal{N}_2\left(0,I_2\right).
\end{equation}
Define $H_n:=\hat{D}_T^{-1}\hat{\Gamma}_T\hat{D}_T^{-1}$ and  
$$M_n:=\left\{\begin{array}{c} 0 \mbox{ if } T/n=o(1),\\ \sqrt{\frac{T}{n}}\hat{D}_T^{-1}\hat{E}_T \mbox{ if } T/n=O(1).\end{array}\right.$$
Let $q_{\beta}$ denote the quantile of order $\beta$ of the standard Gaussian distribution. 
Then the interval 
$$I_{i,n}=\left[\hat{\theta}_i-\frac{\sqrt{H_n(i,i)}q_{1-\alpha/2}+M_n(i)}{\sqrt{nT}},\hat{\theta}_i+\frac{\sqrt{H_n(i,i)}q_{1-\alpha/2}-M_n(i)}{\sqrt{nT}}\right]$$
is a confidence interval with asymptotic level $1-\alpha$ for $i=1,2$. 
\end{cor}
In practice, it is difficult to know which regime is the most appropriate. However, if $T/n=o(1)$, we have $M_n=o(1)$, under our assumptions. Then the bias correction can still be used in this case, though it is asymptotically negligible.  
Moreover, even when $T/n$ seems to be small, the bias correction $M_n$ depends on some empirical averages and can be non-negligible. 
This is why we recommend to use such a bias correction whatever the context, or at least to compute both confidence intervals and to see how they differ from each other.

\begin{figure}[H]
\begin{center}
\includegraphics[width=7cm,height=7cm]{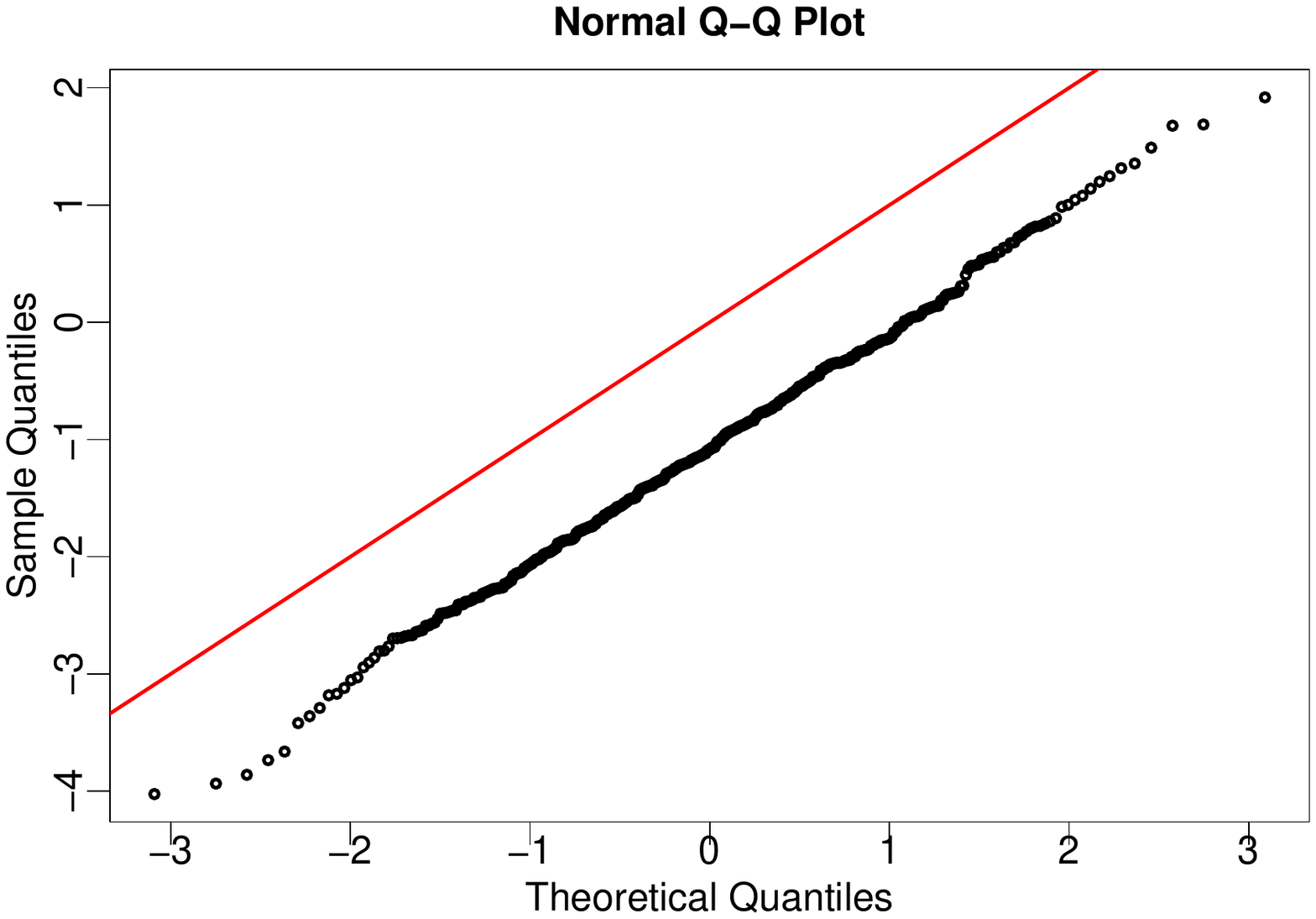}
\includegraphics[width=7cm,height=7cm]{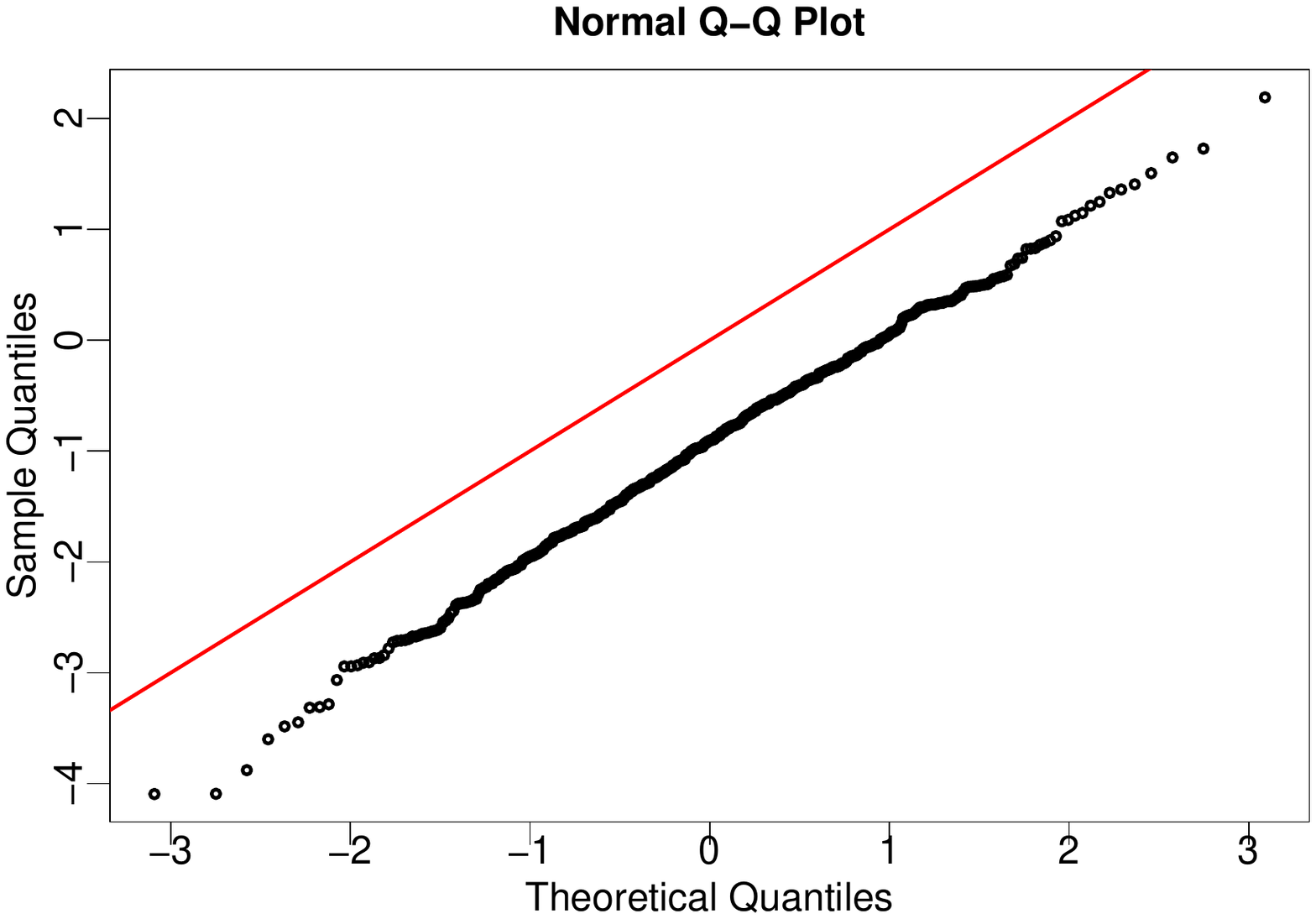}
\end{center}
\begin{center}
\includegraphics[width=7cm,height=7cm]{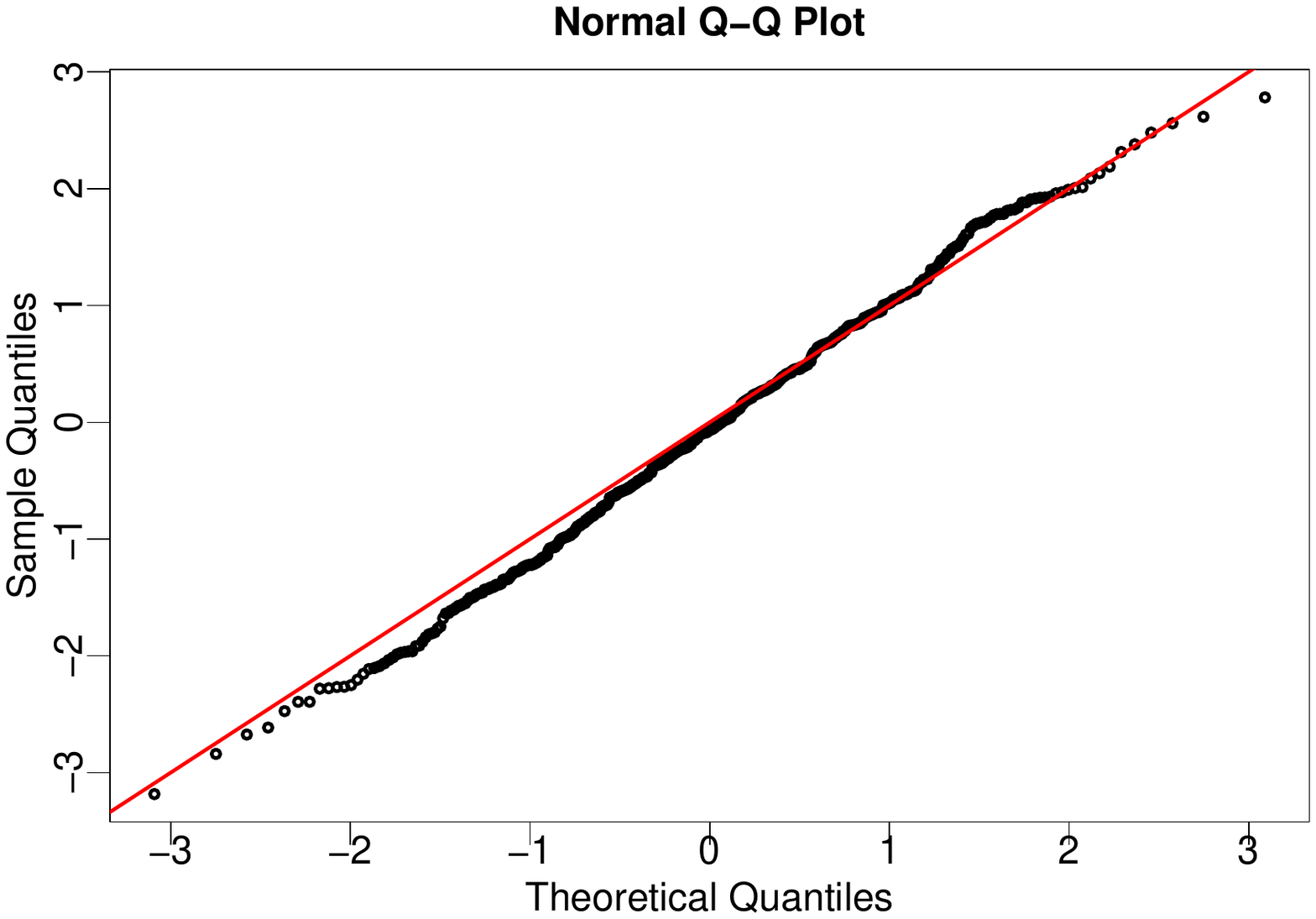}
\includegraphics[width=7cm,height=7cm]{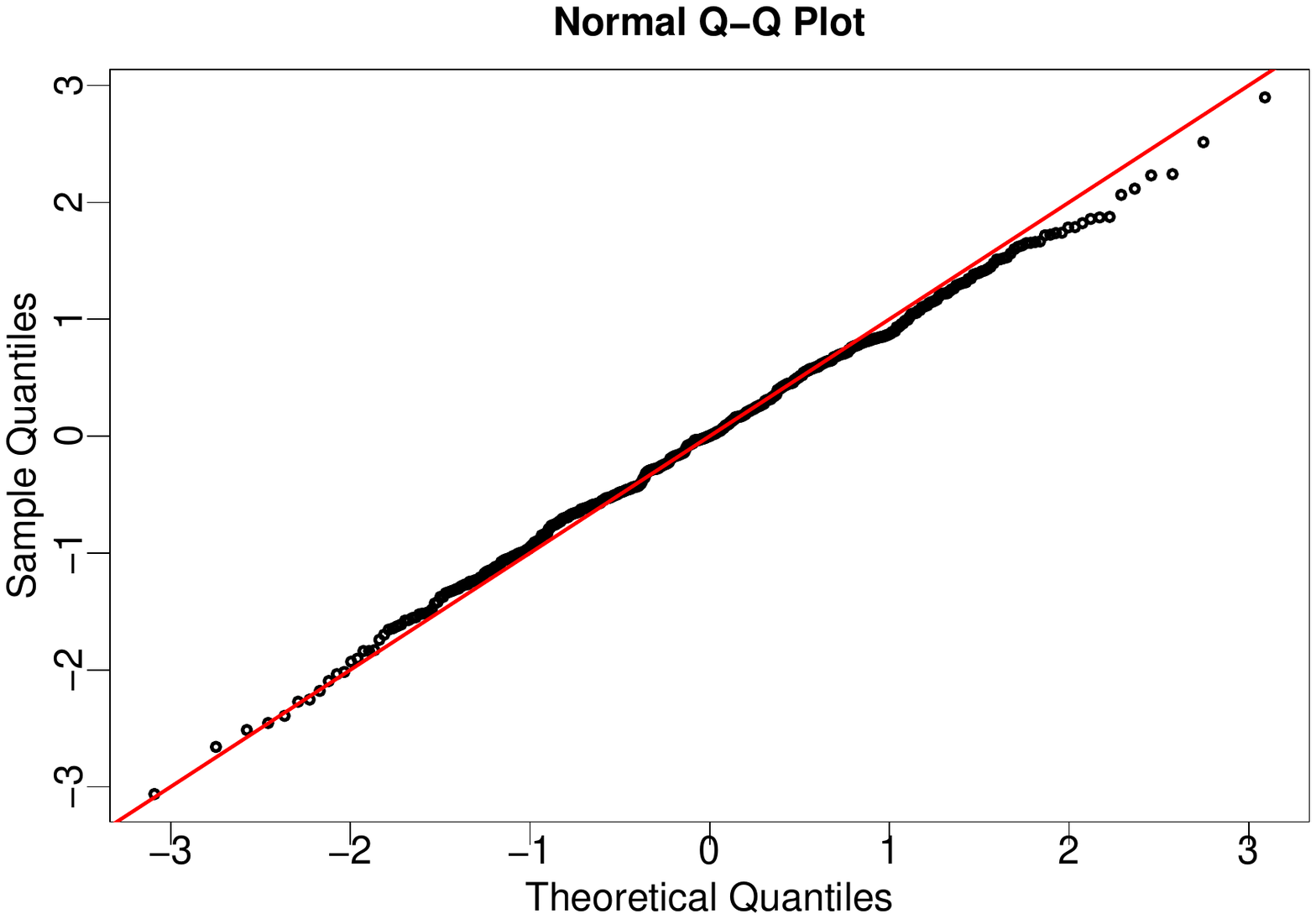}
\end{center}
\caption{For simulated Poisson variates $X_{j,t}\sim \mathcal{P}(\lambda_t)$ with $n=T=100$ and $\lambda_t=3+\cos(t)$, quantile-quantile plots (top) for the two components of statistics $\sqrt{nT}\hat{C}_T\hat{D}_T(\hat{\theta}-\theta)$ and (bottom) for the two components of statistics $\sqrt{nT}\hat{C}_T\hat{D}_T(\hat{\theta}-\theta)-\frac{1}{2}\sqrt{T/n}\hat{C}_T\hat{E}_T$. 
}\label{biasst}
\end{figure}

\section{Example: International Bottom Trawl Survey for the North Sea}\label{sec6}

We analyze some data from the International Council for the Exploration of the Sea (ICES) International Bottom Trawl Survey for the North
Sea (NS-IBTS). 
The survey estimates the abundance, measured as Catch Per Unit Effort (CPUE), of each fish and some invertebrate species 
binned into length classes within spatial rectangles of $0.5^{\circ}$ latitude by $1.0^{\circ}$ longitude,
with an approximate area of $30$ km$^2$, known as subareas. 
Data included in this study are limited to the years $1977$ through $2015$. 
The dataset used here contains only the observations of the nine standard species:
\textit{Clupea harengus} (species $1$), \textit{Gadus morhua} (species $2$), 
\textit{Melanogrammus aeglefinus} (species $3$), \textit{Merlangius merlangus} (species $4$), 
\textit{Pleuronectes platessa} (species $5$), \textit{Pollachius virens} (species $6$), 
\textit{Scomber scombrus} (species $7$), \textit{Sprattus sprattus} (species $8$) and \textit{Trisopterus esmarkii} (species $9$).  
{
The CPUE is given by the number of individuals caught per cruise in $30$ 
minutes of fishing, where only individuals in certain categories of length (measured in mm) were included.}
Further, data were limited to quarter one only, when survey cruises typically occur from mid January to mid February. 
It then excludes seasonal differences in abundance.  
For each of the nine species, the abundance is recorded at $n=195$ locations and across {
$T=39$} years, with one observation per year. See \citet{cobain} and the references therein for more details about this survey.

Figure \ref{graph} gives log-log plots (using the natural logarithm) of the empirical mean and the empirical variance of 
{
the CPUE of} the nine species.

Table \ref{conv} gives the OLS estimate of the slope parameter as well as the $95\%$ confidence bands obtained from standard linear regression methods. In contrast, Table \ref{real} provides confidence bands computed from our asymptotic results, with and without bias correction, for the same confidence level.
We do not consider species $7$ here because more than $75\%$ of these data
were zeros and we had numerical difficulties evaluating the quantities in Corollary \ref{CB}. 
Also, we divided all the abundances by a factor 
{
(the mean value of the abundance of the species, computed over the $n\times T$ observations)}
to avoid numerical issues due to many large values in the time series.
This operation does not change the value of the estimate of the slope 
parameter $\beta$ because TL is still valid with the same $\beta$, 
{
although it changes the value of the intercept parameter $\alpha$}.  

The confidence intervals given in Table \ref{conv} are wider. This is not surprising since their size is proportional to $1/\sqrt{T}$ while the correct size is proportional to $1/\sqrt{nT}$. 
As a consequence, after this correction of the width of the confidence interval, some values of $\beta$ are no longer admissible at the $95\%$ level,  
for instance, the value $\beta=2$ for Species $4$. Recall that the two intervals with or without bias correction have exactly the same length. 
The confidence interval obtained with bias correction is a simple shift, to the left or to the right, of that obtained without bias correction and it is then no longer centered around the estimate $\hat{\beta}$. 

With or without bias correction, the confidence intervals of four species, $1, 5, 8, 9$, include the value $\beta=2$; three intervals lie entirely above $\beta=2$; and two intervals lie entirely below $\beta=2$.
Thus the hypothesis $\beta=2$ is not rejected for only four species, 
{under the important, probably influential assumption, 
for each species, that its time series are independently and identically distributed at different spatial locations.
Our methods (and most other applications of least-squares regression to estimate the parameters of TL)
ignore the possibility that exogenous environmental fluctuations like the North Atlantic Oscillation influence all locations simultaneously, 
as well as the possibility that spatially adjacent or spatially remote populations of a species interact.
		
\begin{figure}[H]
\centering
\begin{subfigure}[b]{0.3\textwidth}
\centering
        \includegraphics[width=4cm,height=4cm]{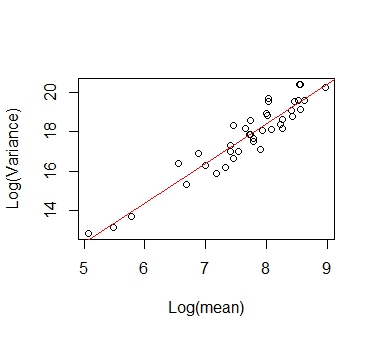}
\caption{Species $1$}
\end{subfigure}
	\begin{subfigure}[b]{0.3\textwidth}
        \centering
				\includegraphics[width=4cm,height=4cm]{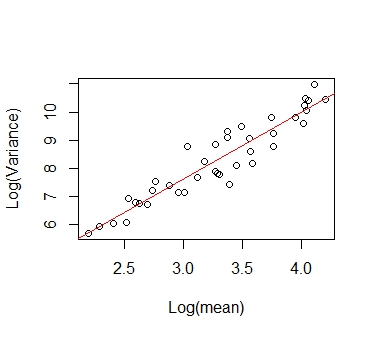}
\caption{Species $2$}
\end{subfigure}
\begin{subfigure}[b]{0.3\textwidth}
        \centering
				\includegraphics[width=4cm,height=4cm]{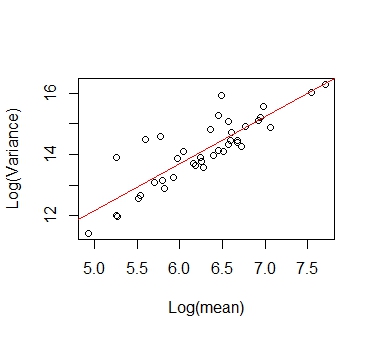}
\caption{Species $3$}
\end{subfigure}
\bigskip

\begin{subfigure}[b]{0.3\textwidth}
\centering
        \includegraphics[width=4cm,height=4cm]{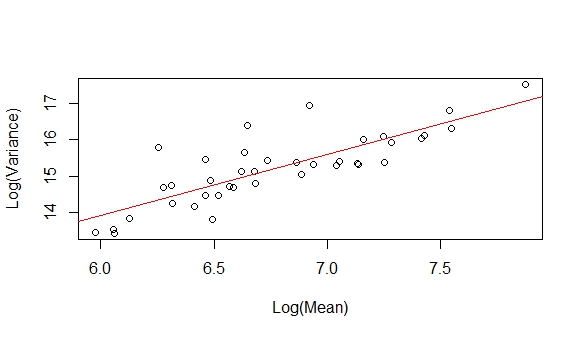}
\caption{Species $4$}
\end{subfigure}
\begin{subfigure}[b]{0.3\textwidth}
\centering
        \includegraphics[width=4cm,height=4cm]{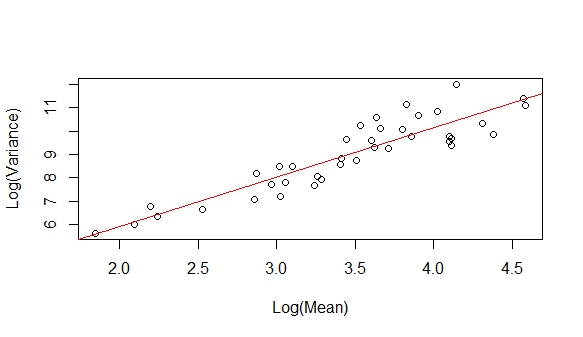}
\caption{Species $5$}
\end{subfigure}
\begin{subfigure}[b]{0.3\textwidth}
\centering
        \includegraphics[width=4cm,height=4cm]{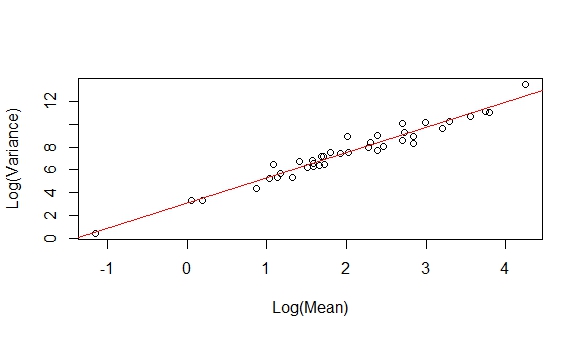}
\caption{Species $6$}
\end{subfigure}
\bigskip

\begin{subfigure}[b]{0.3\textwidth}
\centering
        \includegraphics[width=4cm,height=4cm]{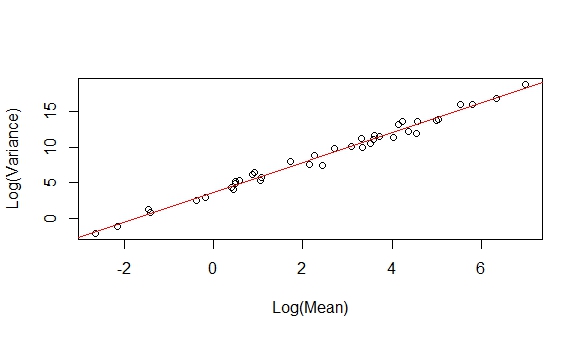}
\caption{Species $7$}
\end{subfigure}			
     \begin{subfigure}[b]{0.3\textwidth}
		\centering
        \includegraphics[width=4cm,height=4cm]{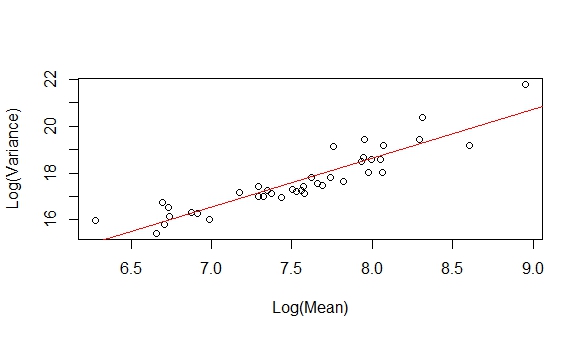}
\caption{Species $8$}
\end{subfigure}
\begin{subfigure}[b]{0.3\textwidth}
\centering
        \includegraphics[width=4cm,height=4cm]{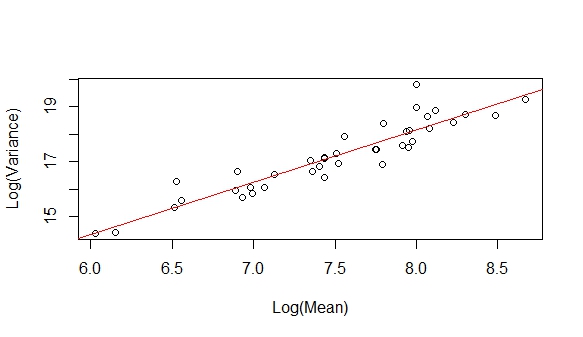}
\caption{Species $9$}
\end{subfigure}
\caption{Plots of the natural logarithm of sample variances versus the natural logarithm of sample means {
of the CPUE of} the nine species. 
The ordinary least-squares regression lines are red. \label{graph}}
\end{figure}

\begin{table}[H]
\begin{center}
\begin{tabular}{|c|c|c|c|}
\hline
Species & OLS $\beta$ & lower bound & upper bound \\\hline
1&2.009793& 1.780779 &2.238807\\\hline
2& 2.390582&2.1058999 &2.675264 \\\hline
3& 1.530209&1.221486 &1.838933\\\hline
4& 1.696704& 1.301093 & 2.092315\\\hline
5& 2.106150& 1.780474 &2.431825\\\hline
6&2.205332&2.041850& 2.368813\\\hline
7&2.09691 &2.015611 &2.178211\\\hline
8&2.084444& 1.779885& 2.389003\\\hline
9&1.893228&1.643332 &2.143124\\\hline
\end{tabular}
\end{center}
\caption{Point estimate by ordinary least squares of the slope parameter $\beta$ for the nine species and lower and upper bound of the $95\%$ confidence interval obtained with the conventional linear regression\label{conv}}
\end{table}

\begin{table}[H]
\begin{center}
\begin{tabular}{|c|c|c|c|c|c|}
\hline
Species & OLS $\beta$ & lower bound NC & upper bound NC& lower bound C& upper bound C\\\hline
1&2.009793& 1.886288& 2.133297&1.918103 &2.165112 \\\hline
2& 2.390582&2.228910 &2.552254& 2.286405& 2.609748\\\hline
3& 1.530209&1.379527& 1.680892&1.346237&1.647602\\\hline
4& 1.696704& 1.429183 &1.964225& 1.356336& 1.891378\\\hline
5& 2.106150& 1.926557 &2.285742&1.965488 &2.324674\\\hline
6&2.205332& 2.073231& 2.337432&2.191456& 2.455656\\\hline
8&2.084444& 1.864855 &2.304033&1.983217& 2.422395\\\hline
9&1.893228&1.735085 &2.051371&1.770439& 2.086725\\\hline
\end{tabular}
\end{center}
\caption{Point estimate by ordinary least squares of the slope parameter $\beta$ for eight of the nine species and lower and upper bound of the $95\%$ confidence interval ($n=195$, {
$T=39$}) 
computed with bias correction (C) and without bias correction (NC)
 \label{real}
}
\end{table}

 Visual inspection of Figure \ref{graph} gives no suggestion, for any species, that log variance is systematically non-linearly related to log mean.
In an exploratory analysis, we fitted a quadratic {
spatial}  {
TL}, 
namely, log spatial variance
$= \log(a) + b\times (\log $ 
{
spatial} mean}) $ +\ c \times (\log $ {
spatial} mean})$^2$,
by ordinary least squares as implemented in the software $R$. 
The null hypothesis that $c=0$ was not rejected for any species, but the confidence intervals for $c$ given by the software are not credible
in the context of this application.
Further theoretical developments are required to assess rigorously whether the data 
reject the null hypothesis of a linear relation between log 
{
spatial}  mean and log 
{
spatial}   variance.
For now, we accept the visual evidence that TL, linear on log-log coordinates, is a plausible model of the variance function of these data.

The present paper does not analyze
stochastic dependence between observations. 
Considering dependence 
will lead to a more delicate estimation of the normalizing matrix $C_T$ in Theorem \ref{NA} and it would probably require large data sets 
to get accurate results.

{To decide if this data set is compatible with the assumption of temporal independence under a nonstationary framework, we compute the residuals $\hat{Z}_{j,t}=\frac{X_{j,t}-\hat{\mu}_t}{\hat{\sigma}_t}$ which behave more like a stationary time series since their two first moments approximate $0$ and $1$. 
The arctanh function has been applied to all the correlation coefficients as a variance-stabilizing function. 
For each species, we then compute the three first autocorrelation coefficients $\rho(1)$, $\rho(2)$, $\rho(3)$ with lags 1, 2, 3 of each time series $(\hat{Z}_{j,t})_{1\leq t\leq T}$ for $j=1,\ldots,n$. 
\citet[Section 3.2]{vdv} defines the corresponding confidence interval of the autocorrelation coefficients, based on this variance-stabilizing function.}

Table \ref{test} gives the percentage of confidence bands that include $0$ for each species. 
The confidence bands were computed under the assumption of independence across $t$ of $Z_{j,t}=\frac{X_{j,t}-\mu_t}{\sigma_t}$ and by neglecting the sampling error in estimating the local mean and local variance. 
Because  $n = 195$ is much larger than 
{
$T = 39$} for this example,
one might hope this sampling error would be negligible. 
However, the results have to be taken with caution since $T$ is not very large. 
For species $6$ to $9$, the assumption of no correlation across time seems to be perfectly plausible. 
Species $3$ seems to be the most problematic.
 
 In Table \ref{test2}, we also apply the same procedure to test spatial correlation. More precisely,  for each of the nine species, 
we compute the sample or empirical correlation $r$ between $Z_{j,t}$ and $Z_{\ell,t}$ for two different locations $j$ and $\ell$ 
using the two samples $(Z_{j,t})_{1\leq t\leq T}$ and $(Z_{\ell,t})_{1\leq t\leq T}$.
We report the percentage of pairs $(j,\ell)$ for which $r=0$ is not rejected.  As for temporal dependence, we note that for some species, such as species $6$, it is difficult to exclude a spatial correlation whereas for some other, spatial correlation seems not to be very important.

Adding an assumption of temporal dependence and or spatial dependence could be an interesting extension of this work. On the other hand, our empirical results suggest that these independence assumptions 
{
are a reasonable first approximation for this data set.} 

{
We have given} a first rigorous basis for estimating the TL parameters 
{
and we have shown that our methods are applicable
and illuminating for a} real data set. 
When extending our work to allow spatial and temporal dependence,
a similar theoretical analysis can be conducted using our proof procedure.
The asymptotics of the ratio $T/n$ will be also of crucial importance in the statement of the corresponding results.

\begin{table}[H]
\begin{center}
\begin{tabular}{|c|c|c|c|}
\hline
Species & $\rho(1)$ & $\rho(2)$& $\rho(3)$\\\hline
1&77.4& 80.0& 99.0\\\hline
2& 75.4& 89.2& 91.3\\\hline
3& 48.2& 51.3& 52.8\\\hline
4& 70.0& 79.0& 82.0\\\hline
5& 81.5& 85.6& 90.3\\\hline
6&97.9& 98.0& 100.0\\\hline
7&95.9& 98.5& 99.0\\\hline
8&97.4& 97.4& 96.9\\\hline
9&95.4& 95.4& 96.9\\\hline
\end{tabular}
\end{center}
\caption{Percentage of confidence intervals containing the point $\rho(i)=0$ for serial autocorrelation with lag $i=1,2,3$. 
The confidence intervals were computed for the {
$n=195$} 
time series using the arctanh variance-stabilizing function. \label{test}}
\end{table}

\begin{table}[H]
\begin{center}
\begin{tabular}{|c|c|c|c|c|c|c|c|c|c|}
\hline
Species & $1$ & $2$& $3$& $4$&$5$&$6$&$7$&$8$&$9$\\\hline
Percentage &  $83$ & $84.5$& $60.4$& $83.8$&$63.4$&$49.4$&$68.1$&$66.7$&$66$\\\hline

\end{tabular}
\end{center}
\caption{Percentage of confidence intervals containing the point $r=0$ for correlation of the abundance between different locations. 
The confidence intervals were computed using the arctanh variance-stabilizing function. \label{test2}}
\end{table}

\section{Conclusions and perspectives}\label{sec7}

In this work, we introduced precise assumptions on the data-generating process 
to study the asymptotic properties of least-squares estimators of Taylor's power law parameters. 
 Theorems \ref{th1}, \ref{th2}, and \ref{NA} are our main new results.
Our assumptions on the marginal distributions are general enough to cover many situations encountered in practice, 
excluding data with fat tails and spatial and/or temporal dependence.

{
To ecologists interested in using our results to control the 
uncertainty in estimates of TL parameters, we offer the following guidance.
\begin{itemize}
    \item 
    Always apply the confidence intervals corresponding to the bias correction given in Corollary \ref{CB}, with $M_n$ defined as if $T/n=O(1)$. This correction is crucial when $T$ and $n$ are of the same magnitude. Moreover, even if the ratio $T/n$ seems to be quite small, it is difficult to discuss a particular threshold to know if the asymptotic regime $T/n=o(1)$ is 
    appropriate. The formula with bias correction always works asymptotically when $T/n=o(1)$.
    Even for finite sample sizes, the confidence bands should not be too much affected.  
    \item 
    The most important distortion of the confidence intervals could come from spatial or temporal dependence that we ignored in this work. 
    We recommend checking whether the independence assumption is compatible with the data set, as we did, for instance, in our real data application. 
    If an important deviation from independence is observed, further research is needed. 
    Our specific formulas for confidence intervals
    may or may not prove to be applicable to data with demonstrable spatial and/or temporal dependence,
    but we believe our approach to modeling and analysis
    of TL parameter estimation generalizes even to data
    with dependence. 
\end{itemize}
}

Some hints for various extensions are given below. When data have fat tails, the limiting distribution of least-squares estimators will probably no longer be Gaussian.  
Even with thin tails, our contribution appears to be the first to derive an asymptotic distribution for such estimates,
which appear in numerous ecological papers without theoretical analysis. 

Various extensions could also be
investigated. For instance, one can study other models
\begin{equation}\label{ondiscute}
\log\sigma_t^2=f_{\theta}(\log\mu_t)
\end{equation}
where $f_{\theta}$ is an unknown function depending on a finite number of parameters,
such as the quadratic and the cubic TL mentioned in the introduction.
Another extension could use a statistical learning point of view by studying more general least-squares problems of the form 
\begin{equation}\label{ondiscute2}
\hat{f}=\underset{f\in\mathcal{F}}{\arg\min}\sum_{t=1}^T\left(\log\hat{\sigma}_t^2-f(\log\hat{\mu}_t)\right)^2,
\end{equation}
where $\mathcal{F}$ is an infinite-dimensional class of functions. 
Of course, we have to take care that whatever the regularity (e.g. differentiability) imposed on $\mathcal{F}$, 
there will be an infinite number of functions mapping the sequence $(\log\mu_t)_t$ onto the sequence $(\log\sigma_t^2)_t$.
This non-parametric extension \eqref{ondiscute2} will differ from the usual standard non-parametric estimation problems in statistics.
Testing a specific parametric model of the form (\ref{ondiscute}) could be also interesting.

{
Another open question related to Assumption {\bf A1 } is  
to check whether time and space ($t$ and $j$) could be exchanged.
Non-stationarity in time seems mandatory to get identifiable parameters $\alpha, \beta$ in \eqref{main}. 
Estimating the normalizing quantities could be more difficult in this nonstationary context.}
Taking account of the stochastic dependence, spatial or temporal or both, between the observations is also a technical challenge. 
Many ways to extend Theorem \ref{NA} in the context of dependence have been published;
for instance, \citet{dedecker2007weak} provide many helpful asymptotic results.

Extensions to dependent random variables are plausible in different settings. For example, 
\cite{delapena2022taylor} suggest a variant definition of TL together with a simple extension of the TL under mixing type conditions.  
\cite{cohen2020heavy} consider heavy-tailed processes.    
\cite{taqqu2002lrd} give a systematic approach to a
weaker condition such as long-range dependence.
\cite{samo2007lrd} gives a comprehensive approach to long-range dependence under heavy-tailed assumptions.
 We plan to address these open questions in future work.

\section{Proof of the results}\label{sec8}

Here we prove our main results. Some proofs use technical lemmas in the Appendix. 
The proof is difficult because
we did not find adequate a.s. convergence results adapted to the current triangular arrays setting. This led us to a systematic use of the Borel-Cantelli lemma and the bounds for the probability of the events on which the estimates do not have a satisfying behaviour. Those quantities are controlled through standard exponential inequalities or through the Markov inequality and suitable  moment inequalities. From now on, we denote by $\vert\cdot\vert$ the Euclidean norm on $\R^2$: $|(x, y)| := (x^2+y^2)^{1/2}$.
We also denote the corresponding operator norm on the set of $2\times 2$ matrices by $\vert\cdot\vert$, namely, $|A|:=\sup\{|Az|/|z|\mid z\in \R^2\}$. Finally, if $X$ is a random variable, we denote by $\Vert X\Vert_p=\E^{1/p}\left[\vert X\vert^p\right]$ the $\L_p$-norm of $X$, for any $p\geq 1$.

\subsection{Proof of Theorem \ref{th1}.}

The proof {
will be based on the technical Lemma \ref{prems} in the Appendix.  
Lemma \ref{prems} controls  the probability of getting small values for empirical means and empirical variances. 
We then use the constants depending on $\epsilon$ and $L$ defined in  Lemma \ref{prems}.} 
\begin{proof}
Set $\gamma\in (0,L)$. 
Define the following event 
$$E_n(\epsilon):=\bigcap_{t=1}^T\left\{\hat{m}_t>\epsilon,\hat{\sigma}_t^2>\epsilon\right\}.$$
This sequence of events depending on $n$ takes care of the properties of the estimates
 when they take small values. Using the union bound and Lemma \ref{prems}, we have
$$\P\left(E_n(\epsilon)^c\right)\leq 2T\exp(-nL).$$
This shows that $\sum_{n\geq 2}\P(E_n(\epsilon)^c)<\infty$.

We now study the consistency of $\hat{\theta}$. 
We only need to show that $\hat{D}_T-D_T$ and $\hat{N}_T-N_T$ go to $0$ in probability or a.s. 
Indeed, if we use the decomposition 
$$\hat{\theta}-\theta=\hat{D}_T^{-1}\left(\hat{N}_T-N_T+(D_T-\hat{D}_T)D_T^{-1}N_T\right),$$
Assumptions {\bf A2} and {\bf A4} guarantee that $N_T=O(1)$ and $D_T^{-1}=O(1)$. 
There then also exists $\delta>0$ sufficiently small such that
$\sup_{\left\vert H\right\vert\leq \delta}\left\vert (D_T+H)^{-1}-D_T^{-1}\right\vert=O(1)$. 

We will then conclude that $\hat{D}_T^{-1}-D_T^{-1}$ converges to $0$ in probability or a.s. 
The weak or strong consistency of $\hat{\theta}$ will then follow.

Now, for some $\eta>0$, we have
\begin{equation}\label{objectif} 
\P(\vert \hat{\theta}-\theta\vert>\eta)\leq 
\P(\{\vert \hat{\theta}-\theta\vert>\eta\}\cap E_n(\epsilon))+\P\left(E_n(\epsilon)^c\right).
\end{equation}
We will use the following inequality extensively. 
If $a,b\geq \epsilon$ and $\kappa\in (0,1)$, we have 
$$\left\vert \log a-\log b\right\vert=\frac{1}{\kappa}\left\vert \log a^{\kappa}-\log b^{\kappa}\right\vert\leq \frac{1}{\kappa \epsilon^{\kappa}}\left\vert a-b\right\vert^{\kappa}.$$
 This inequality turns the log-problems across the origin into easier moment bounds. 
To get this inequality, one can apply the mean value theorem, then use the inequality $\vert a^{\kappa}-b^{\kappa}\vert\leq \vert a-b\vert^{\kappa}$ for $\kappa\in (0,1)$.

On $E_n(\epsilon)$, we get
$$\left\vert \log \hat{\sigma}^2_t-\log\sigma_t^2\right\vert\leq \frac{1}{\kappa \epsilon^{\kappa}}\left\vert \hat{\sigma}_t^2-\sigma_t^2\right\vert^{\kappa}.$$

From Assumption {\bf A2} 
and standard moment inequalities for partial sums (Lemma \ref{moment}), we have 
$$\sup_{t\geq 1}\E\left\vert \hat{\sigma}_t^2-\sigma_t^2\right\vert^{\kappa}\leq\sup_{t\geq 1}\E^{2\kappa/k}\left\vert \hat{\sigma}_t^2-\sigma_t^2\right\vert^{k/2}=o(1).$$
By choosing $\kappa$ small enough in the previous bounds, for any positive integer $i$, we get
\begin{equation}\label{interm1}
\sup_{t\geq 1}\E\left[\left\vert \log\hat{\sigma}^2_t-\log\sigma^2_t
\right\vert^i
\mathds{1}_{E_n(\epsilon)}\right]=o(1)
\end{equation}
and similarly 
\begin{equation}
\label{interm2}
\sup_{t\geq 1}\E\left[\left\vert \log \hat{\mu}_t-\log \mu_t\right\vert^i\mathds{1}_{E_n(\epsilon)}\right]=o(1).
\end{equation}
On the event $E_n(\epsilon)$,  
$T^{-1}\sum_{t=1}^T\left[\log \hat{\sigma}^2_t-\log\sigma^2_t\right]\to 0$ in $\L^1$ and then in probability when $n\rightarrow \infty$.
Moreover, setting $C:=\max\left(\sup_{t\geq 1}\left\vert\log \sigma_t^2\right\vert,\sup_{t\geq 1}\left\vert\log \mu_t\right\vert\right)$, we have
\begin{align*}
\left\vert \log  \hat{\mu}_t\log\hat{\sigma}^2_t-\log \mu_t\log \sigma^2_t\right\vert
\end{align*}
\begin{align*}
\leq& \left\vert \log \hat{\mu}_t-\log \mu_t\right\vert\cdot\left\vert \log \hat{\sigma}^2_t-\log \sigma^2_t\right\vert+\left\vert\log \sigma^2_t\right\vert\cdot \left\vert \log \hat{\mu}_t-\log \mu_t\right\vert
+ \left\vert \log \mu_t\right\vert\cdot\left\vert\log \hat{\sigma}^2_t-\log \sigma^2_t\right\vert\\
\leq& \frac{1}{2}\left\vert  \log \hat{\sigma}^2_t-\log \sigma^2_t\right\vert^2+\frac{1}{2}\left\vert \log \hat{\mu}_t-\log \mu_t\right\vert^2+C\left\vert \log \hat{\mu}_t-\log \mu_t\right\vert+C\left\vert \log \hat{\sigma}^2_t-\log \sigma^2_t\right\vert.
\end{align*}
From (\ref{interm1}) and (\ref{interm2}), we also deduce that 
$$\frac{1}{T}\sum_{t=1}^T\E \mathds{1}_{E_n(\epsilon)}\left\vert  \log \hat{\mu}_t\log\hat{\sigma}^2_t-\log \mu_t\log \sigma^2_t\right\vert=o(1).$$
Hence $\hat{N}_T-N_T \to 0$ in probability on the event $E_n(\epsilon)$. Similarly, $\hat{D}_T-D_T\to 0$ in probability on the event $E_n(\epsilon)$. This leads to the weak consistency of the regression estimates.

{
For any $\eta>0$,
$$\sum_{n\geq 1}\P\left(\left\{\left\vert \hat{\theta}-\theta\right\vert>\eta\right\}\cap E_n(\epsilon)^c\right)\leq \sum_{n\geq 1}\P\left(E_n(\epsilon)^c\right)<\infty$$
and the Borel-Cantelli lemma ensures that $\left\vert \hat{\theta}-\theta\right\vert\cdot\mathds{1}_{E_n(\epsilon)^c}$ converges to $0$ a.s.
}
To prove the almost sure convergence, it is enough to prove that $\hat{N}_T-N_T\to 0$ 
and $\hat{D}_T-D_T\to 0$ a.s. on the event $E_n(\epsilon)$. 
Using the same strategy as for the convergence in probability, {
it suffices} to show that
\begin{equation}\label{supinterest}
\lim_{T\rightarrow \infty}
\frac{1}{T}\sum_{t=1}^T\left\{\left\vert \hat{\sigma}^2_t-\sigma^2_t\right\vert+\left\vert \hat{\mu}_t-\mu_t\right\vert\right\}= 0 \quad\mbox{ a.s.}
\end{equation}
Because the function $x\mapsto \vert x\vert^{k/2}$ (where $k>2$ is given in Assumption {\bf A2} is convex,
we get
$$M_T:=\left(\frac{1}{T}\sum_{t=1}^T\left\vert\hat{\sigma}^2_t-\sigma^2_t\right\vert\right)^{k/2}\leq \frac{1}{T}\sum_{t=1}^T \left\vert \hat{\sigma}^2_t-\sigma_t^2\right\vert^{k/2}.$$
From moment inequalities for partial sums of independent observations (Lemma \ref{moment}), 
we get $\E M_T=O(n^{-k/4})$ and by assumption, if $k>4$, we get
 $\sum_{n\geq 1} \E M_T<\infty$.
Setting $P_T=\left\vert \hat{\mu}_t-\mu_t\right\vert$, 
similar arguments yield $\sum_{n\geq 1}\E P_T<\infty$. The Borel-Cantelli lemma ensures that both $M_T$ and $P_T$ go to zero a.s. as $n\rightarrow \infty$ and then the validity of (\ref{supinterest}), {
which concludes the proof}. 
\end{proof}

\paragraph{Proof of Theorem \ref{th2}.}

The proof is similar to that of Theorem \ref{th1}.
The main difference is in proving the convergence of $\hat{D}_T-D_T$ or $\hat{N}_T-N_T$ to $0$.
\begin{proof}
Let $A_t=\left\{\hat{\mu}_t>\epsilon, \hat{\sigma}_t^2>\epsilon\right\}$ where $\epsilon{\blue <1}$ satisfies the condition 
in Lemma \ref{prems}.2.
Using arguments similar to those in the proof of Theorem \ref{th1}, one can show that 
$$\frac{1}{T}\sum_{t=1}^T\mathds{1}_{A_t}\left[\left\vert\log\hat{\sigma}^2_t-\log\sigma_t^2\right\vert+\left\vert \log\hat{\mu}_t-\log\mu_t\right\vert+\left\vert \log^2\hat{\mu}_t-\log^2\mu_t\right\vert
+\left\vert\log\hat{\mu}_t\log\hat{\sigma}^2_t-\log\mu_t\log\sigma_t^2\right\vert\right]$$
converges to $0$ either in probability or a.s., depending on $k$. 
From Lemma \ref{prems},
$$\sum_{n\geq 1}\frac{1}{T}\sum_{t=1}^T\P\left(A_t^c\right)<\infty$$
and from the Borel-Cantelli lemma, we deduce that $\lim_{n\rightarrow \infty}\frac{1}{T}\sum_{t=1}^T\mathds{1}_{A_t^c}=0$ a.s.
Moreover, from Lemma \ref{prems}, the sequences $(\mu_t)_{t\geq 1}$ 
and $(\sigma^2_t)_{t\geq 1}$ are upper- and lower-bounded. 
To show the required convergence, it suffices to show that 
the following quantities converge  a.s. to $0$:
\begin{eqnarray*}
\frac{1}{T}\sum_{t=1}^T \mathds{1}_{0<\hat{\sigma}^2_t<\epsilon}\cdot\left\vert\log\hat{\sigma}^2_t\right\vert,&\quad &\frac{1}{T}\sum_{t=1}^T \mathds{1}_{0<\hat{\sigma}^2_t<\epsilon}\cdot\left\vert\log\hat{\sigma}^2_t \cdot\log\hat{\mu}_t\right\vert,\\
\frac{1}{T}\sum_{t=1}^T \mathds{1}_{0<\hat{\mu}_t<\epsilon}\cdot\left\vert\log\hat{\sigma}^2_t\cdot\log\hat{\mu}_t\right\vert,&\quad &\frac{1}{T}\sum_{t=1}^T \mathds{1}_{0<\hat{\mu}_t<\epsilon}\cdot\left\vert\log\hat{\mu}_t\right\vert,\\
\frac{1}{T}\sum_{t=1}^T \mathds{1}_{0<\hat{\mu}_t<\epsilon}\cdot\left\vert\log^2\hat{\mu}_t\right\vert.&
\end{eqnarray*}
We recall that $\log\hat{\mu}_t$ and $\log\hat{\sigma}^2_t$ are implicitly defined respectively by $\log \hat{\mu}_t\cdot\mathds{1}_{\hat{\mu}_t>0}$ and $\log \hat{\sigma}^2_t\cdot \mathds{1}_{\hat{\sigma}^2_t>0}$.

First, {
because $\log(\epsilon)<0$}, we have that 
$$\log\hat{\mu}_t\cdot\mathds{1}_{0<\hat{\mu}_t<\epsilon}=\log\hat{\mu}_t\cdot\mathds{1}_{0<\hat{\mu}_t<\epsilon}\cdot\mathds{1}_{\cup_{j=1}^n\{X_{j,t}>0\}}\geq \sum_{j=1}^n \log\left(X_{j,t}/n\right)\cdot\mathds{1}_{X_{j,t}>0}\cdot\mathds{1}_{0<\hat{\mu}_t<\epsilon}.$$
We then get the bound
\begin{equation}\label{onyva}
\left\vert \log \hat{\mu}_t\right\vert\cdot\mathds{1}_{\hat{\mu}_t>0}\leq \left\vert \log \hat{\mu}_t\right\vert\cdot\mathds{1}_{\hat{\mu}_t>\epsilon}
-\sum_{j=1}^n\log(X_{j,t}/n)\cdot\mathds{1}_{X_{j,t}>0}\cdot\mathds{1}_{0<\hat{\mu}_t<\epsilon}.
\end{equation}
Using the same type of arguments for
$$\hat{\sigma}^2_t=\frac{1}{2n^2}\sum_{1\leq i\neq j\leq n}\left\vert X_{i,t}-X_{j,t}\right\vert^2,$$
we get the bound 
\begin{equation}\label{onyvava}
\left\vert \log \hat{\sigma}^2_t\right\vert\cdot\mathds{1}_{\hat{\sigma}^2_t>0}\leq \left\vert \log \hat{\sigma}^2_t\right\vert\cdot\mathds{1}_{\hat{\sigma}^2_t>\epsilon}
-\sum_{1\leq i\neq j\leq n}\log\left(\frac{\vert X_{i,t}-X_{j,t}\vert^2}{2n^2}\right)\cdot\mathds{1}_{X_{i,t}\neq X_{j,t}}\cdot\mathds{1}_{0<\hat{\sigma}^2_t<\epsilon}.
\end{equation}
We are going to show that 
\begin{equation}\label{suffit}
Z_T:=\frac{1}{T}\sum_{t=1}^T \mathds{1}_{0<\hat{\sigma}^2_t<\epsilon}\cdot\log\hat{\sigma}^2_t\cdot\log\hat{\mu}_t\rightarrow_{n\rightarrow \infty}0
\end{equation}
in probability or a.s.
The idea for all the other quantities is exactly the same and details will be omitted. 

Let $p=q/2>1$ by Assumption {\bf A5} and $\alpha=1/(2p)$. From (\ref{onyva}), there exists $C_q>0$ such that

$$\Vert\log\hat{\mu}_t\cdot\mathds{1}_{\hat{\mu}_t>0}\Vert_{2p}\leq C_q\cdot\E^{\alpha}\left(\hat{\mu}_t\right)+\sum_{j=1}^n\Vert \log \frac{X_{j,t}}{n}\mathds{1}_{X_{j,t}>0}\Vert_{2p}.$$
From {\bf A5}, we deduce the existence of a constant $D>0$ 
such that {
the behaviour of $\hat \mu_t$ around the origin is not too explosive:} 
\begin{equation}\label{meil}
\Vert \log\hat{\mu}_t\cdot\mathds{1}_{\hat{\mu}_t>0}\Vert_{2p}\leq D\left(1+n\log(n)\right).
\end{equation}
From (\ref{onyvava}) and the H\"older inequality, we deduce that 
\begin{multline*}
    \Vert \log\hat{\mu}_t\cdot\log\hat{\sigma}_t^2\cdot \mathds{1}_{0<\hat{\sigma}^2_t<\epsilon}\Vert_1
    \\
    \leq
\sum_{1\leq i\neq j\leq n}\P\left(\hat{\sigma}^2_t<\epsilon\right)^{1-1/p}\cdot\Vert \log\hat{\mu}_t\cdot\mathds{1}_{\hat{\mu}_t>0}\Vert_{2p}\cdot \Vert \log\left(\frac{\vert X_{i,t}-X_{j,t}\vert^2}{2n^2}\right)\cdot\mathds{1}_{X_{i,t}\neq X_{j,t}}\Vert_{2p}.
\end{multline*}
From Assumption {\bf A5}, inequality (\ref{meil}) and Lemma \ref{prems}, it is not very difficult to conclude that 
$$\sum_{n\geq 1}\frac{1}{T}\sum_{t=1}^T\Vert \log\hat{\mu}_t\cdot\log\hat{\sigma}_t^2\cdot \mathds{1}_{0<\hat{\sigma}^2_t<\epsilon}\Vert_1<\infty.$$
From the Borel-Cantelli lemma, we get (\ref{suffit}).
Remembering the discussion at the beginning of the proof, one can deduce the consistency properties of $\hat{\theta}$. 
\end{proof}

\subsection{Proof of Theorem \ref{NA}.}

\begin{proof}
The second derivatives of the function $\phi_{\theta}=(\phi_{1,\theta},\phi_{2,\theta})$ are:
$$\frac{\partial^2\phi_{1,\theta}}{\partial x^2}(x,y)=-\frac{2}{y-x^2}-\frac{4x^2}{(y-x^2)^2}+\frac{\theta_2}{x^2},$$
$$\frac{\partial^2\phi_{1,\theta}}{\partial y^2}(x,y)=\frac{-1}{(y-x^2)^2},\quad\frac{\partial^2\phi_{1,\theta}}{\partial x \partial y}(x,y)=\frac{2x}{(y-x^2)^2},$$
$$\frac{\partial^2 \phi_{2,\theta}}{\partial x^2}(x,y)=\frac{-4}{y-x^2}-\frac{\log(y-x^2)}{x^2}-\frac{2\log x}{y-x^2}-\frac{4x^2\log(x)}{(y-x^2)^2}+\frac{\theta_1-2\theta_2}{x^2}+\frac{2\theta_2\log(x)}{x^2},$$
$$\frac{\partial^2 \phi_{2,\theta}}{\partial y^2}(x,y)=-\frac{\log(x)}{(y-x^2)^2},\quad \frac{\partial^{2}\phi_{2,\theta}}{\partial x\partial y}(x,y)=\frac{1}{x(y-x^2)}+\frac{2x\log(x)}{(y-x^2)^2}.$$
To prove Theorem \ref{NA}.1, we first show that
\begin{equation}\label{interest}
\sqrt{nT}\left(U_T-V_T\right)=o_{\P}(1).
\end{equation}
Using a Taylor-Lagrange expansion of order $2$, we have
\begin{eqnarray}\label{decomp}
\sqrt{nT}(U_T-V_T)
=\left(\frac{1}{2}\frac{\sqrt{n}}{\sqrt{T}}\sum_{t=1}^T(\hat{m}_t-m_t)'H_{1,t}(\hat{m}_t-m_t),\frac{1}{2}\frac{\sqrt{n}}{\sqrt{T}}\sum_{t=1}^T(\hat{m}_t-m_t)'H_{2,t}(\hat{m}_t-m_t)\right),
\end{eqnarray}
where 
$$H_{1,t}=\nabla^{(2)}\phi_{1,\theta}(\widetilde{m}_t),\quad H_{2,t}=\nabla^{(2)}\phi_{2,\theta}(\check{m}_t),$$
for some random points $\widetilde{m}_t, \check{m}_t$ in the segment $[m_t,\hat{m}_t]$. 
It is only necessary to show (\ref{interest}) on the event 
$$E(\epsilon)=\cap_{t=1}^T\left\{\hat{\mu}_t>\epsilon,\hat{\sigma}_t^2>\epsilon\right\},$$
which has been defined in the proof of Theorem \ref{th1} and which satisfies $\P\left(E(\epsilon)^c\right)\leq 2T\exp(-nL)$ for some $L>0$.
From Lemma \ref{prems}.1, on $E(\epsilon)$, we have, setting $w=\min(\epsilon,s)$, 
$$\min_{1\leq t\leq T}\min\left(\mu_t,m_{2,t},\sigma^2_t,\hat{\mu}_t,\hat{m}_{2,t},\hat{\sigma}^2_t\right)\geq w.$$
Moreover, on the same event $E(\epsilon)$, using convexity of the square function, we have for $1\leq t\leq T$,
$$\min\left(\widetilde{m}_{2,t}-\widetilde{m}_{1,t}^2,\check{m}_{2,t}-\check{m}_{1,t}^2\right)\geq \min\left(\hat{\sigma}^2_t,\sigma_t^2\right)\geq w.$$
Also $\min\left(\widetilde{m}_{1,t},\check{m}_{1,t}\right)\geq w$. 
Moreover, there exists a positive constant $\widetilde{C}$ 
such that for $1\leq t\leq T$ and $i=1,2$,
$$\max\left(\widetilde{m}_{i,t},\check{m}_{i,t}\right)\leq \widetilde{C}\left(1+\vert \hat{m}_{i,t}-m_{i,t}\vert\right).$$
It is then sufficient to show that for $i,i'\in \{1,2\}$,
\begin{equation}\label{but1}
\frac{\sqrt{n}}{\sqrt{T}}\sum_{t=1}^Th_{i,i',t}\vert \hat{m}_{i,t}-m_{i,t}\vert\cdot \left\vert \hat{m}_{i',t}-m_{i',t}\right\vert=o_{\P}(1),
\end{equation}
with
$$h_{i,i',t}=\max_{\ell=1,2}\sup_{\substack{w\leq x\leq \widetilde{C}(1+\vert \hat{\mu}_t-\mu_t\vert)\\w\leq z\leq \widetilde{C}(1+\vert \hat{m}_{2,t}-m_{2,t}\vert)}}\left\vert \frac{\partial^2 \phi_{\ell,\theta}}{\partial x^i\partial y^{i'}}(x,z+x^2)\right\vert.$$
A careful inspection of the second partial derivatives of $\phi_{\theta}$ shows that for any $\delta>0$, there exists a constant $C_{\delta}'>0$ such that 
$$h_{1,1,t}\leq C'_{\delta}\left[1+\vert \hat{\mu}_t-\mu_t\vert^{2+\delta}+\vert \hat{m}_{2,t}-m_{2,t}\vert^{\delta}\right],$$
$$h_{2,2,t}\leq C'_{\delta}\left[1+\vert \hat{\mu}_t-\mu_t\vert^{\delta}\right],\quad h_{1,2,t}\leq C_{\delta}'\left[1+\vert\hat{\mu}_t-\mu_t\vert^{1+\delta}\right].$$
Take $\delta\in (0,1)$ such that $4+\delta<k$. Using the moment bound given by Lemma \ref{moment},
it is not difficult to get that 
$$\frac{\sqrt{n}}{\sqrt{T}}\sum_{t=1}^Th_{i,i',t}\vert \hat{m}_{i,t}-m_{i,t}\vert\cdot \left\vert \hat{m}_{i',t}-m_{i',t}\right\vert=O_{\P}(\sqrt{T/n})$$
which yields (\ref{but1}) and then (\ref{interest}).

It now remains to show that
\begin{equation}\label{limit}
W_T:=\sqrt{nT}C_T V_T\Rightarrow \mathcal{N}_2(0,I_2).
\end{equation}
{
We recall that $I_2$ denotes the $2\times 2$ identity matrix.}

This will follow from the Lindeberg-Feller central limit theorem for triangular arrays
\citep[Theorem 18.1]{billingsley2013convergence}. 
Since $\v(W_T)=I_2$, 
we only have to check the Lindeberg condition. Setting $Z_{n,j,t}=\frac{1}{\sqrt{nT}}C_T J_{\phi_{\theta}}(m_t)(Y_{j,t}-m_t)$ and using our assumptions, we have 
$C_T=O(1)$ and $\max_{1\leq t\leq T}\vert J_{\phi_{\theta}}(m_t)\vert=O(1)$. 
For $\delta>0$ such that $2(2+\delta)\leq k$, we get $\sum_{j=1}^n\sum_{t=1}^T\E\vert Z_{n,j,t}\vert^{2+\delta}=O\left((nT)^{-\delta/2}\right)=o(1)$. This proves (\ref{limit}).

To prove Theorem \ref{NA}.2, we again use the decomposition (\ref{decomp}) and we first show that for $\ell=1,2$,
\begin{equation}\label{goal}
\frac{\sqrt{n}}{\sqrt{T}}\sum_{t=1}^T\left(\hat{m}_t-m_t\right)'\left[H_{\ell,t}-K_{\ell,t}\right]\left(\hat{m}_t-m_t\right)=o_{\P}(1),
\end{equation}
where for $\ell=1,2$ and $t\geq 1$, $K_{\ell,t}=\nabla^{(2)}\phi_{\ell,\theta}(m_t)$.
It is only necessary to show (\ref{goal}) on the event $E_n(\epsilon)$. To this end, we will use the bounds 
$$\max\left\{\left\vert \widetilde{\mu}_t-\mu_t\right\vert, \left\vert \check{\mu}_t-\mu_t\right\vert\right\}\leq \left\vert \hat{\mu}_t-\mu_t\right\vert,$$
and 
$$\max\left\{\left\vert \widetilde{\sigma}^2_t-\sigma^2_t\right\vert, \left\vert \check{\sigma}^2_t-\sigma^2_t\right\vert\right\}\leq \left\vert \hat{m}_{2,t}-m_{2,t}\right\vert+\left(2\sup_{s\geq 1}\mu_s +\left\vert \hat{\mu}_t-\mu_t\right\vert\right)\cdot\left\vert \hat{\mu}_t-\mu_t\right\vert,$$
where we set $\widetilde{\sigma}^2_t:=\widetilde{m}_{2,t}-\widetilde{m}_{1,t}^2$ and $\check{\sigma}^2_t:=\check{m}_{2,t}-\check{m}_{1,t}^2$.
We will also use the fact that if $h:(\epsilon,\infty)\rightarrow \R_+$ is any mapping of the form $h(z)=1/z$, $h(z)=1/z^2$ or $h(z)=\log(z)$, then for any $\kappa\in (0,1)$, there exists 
$C_{\kappa}>0$ such that 
$$\left\vert h(z)-h(z')\right\vert\leq C_{\kappa}\vert z-z'\vert^{\kappa}.$$
Assertion (\ref{goal}) can be obtained from these facts and a careful analysis of the second derivatives of $\phi_{\theta}$. 
We will illustrate this by showing that 
\begin{equation}\label{petit}
\frac{\sqrt{n}}{\sqrt{T}}\sum_{t=1}^T\left(\hat{m}_{2,t}-m_{2,t}\right)^2\left[\frac{\partial^{2}\phi_{1,\theta}}{\partial y^2}(\widetilde{m}_t)-\frac{\partial^{2}\phi_{1,\theta}}{\partial y^2}(m_t)\right]=o_{\P}(1).
\end{equation}
For $\kappa\in (0,1)$ small enough to satisfy $4+2\kappa\leq k$, there exists $D_{\kappa}>0$ such that 
$$\left\vert \frac{\widetilde{\mu}_t}{\widetilde{\sigma}^4_t}-\frac{\mu_t}{\sigma_t^4}\right\vert \leq C\left[\left\vert \widetilde{\mu}_t-\mu_t\right\vert+\left\vert \widetilde{\sigma}^2_t-\sigma_t^2\right\vert^{\kappa}\right].$$
To show (\ref{petit}), it is sufficient to show that 
$$\E\left[\left(\hat{m}_{2,t}-m_{2,t}\right)^2\left\vert \hat{m}_{i,t}-m_{i,t}\right\vert^{\kappa_i}\right]=O\left(n^{-1-\frac{\kappa_i}{2}}\right),$$
when $(i,\kappa_i)=(2,\kappa), (1,\kappa)$ or $(1,2\kappa)$. But this point easily follows from Lemma \ref{moment} and H\"older's inequality.
Assertions similar to (\ref{petit}) can be obtained in the same way for the other second-order partial derivatives of $\phi_{\theta}$.
Details are omitted. Then
\begin{eqnarray*}
\sqrt{nT}(U_T-V_T)
&=&\left(\frac{\sqrt{n}}{\sqrt{T}}\sum_{t=1}^T\left(\hat{m}_t-m_t\right)'K_{1,t}\left(\hat{m}_t-m_t\right),\frac{\sqrt{n}}{\sqrt{T}}\sum_{t=1}^T\left(\hat{m}_t-m_t\right)'K_{2,t}\left(\hat{m}_t-m_t\right)\right)\\
&+& o_{\P}(1).
\end{eqnarray*}
One applies Lemma \ref{intern} to conclude the proof. 
\end{proof}

\subsection{Proof of Corollary \ref{CB}.}

The result will be consequence of Theorem \ref{NA} and of Lemmas \ref{conver} and \ref{convuni2} in the Appendix.

\begin{proof} Set $\hat{A}_n:=\hat{C}_T\hat{D}_T$ and $A_n:=C_T \hat{D}_T$.
Using Lemma \ref{conver} in the Appendix, we have $\hat{C}_T-C_T=o_{\P}(1)$ and from the proof of Theorem \ref{th1}, we know that $\hat{D}_T-D_T=o_{\P}(1)$ {
and both $D_T$ and $D_T^{-1}$ are $O(1)$}.
We then easily conclude that $\hat{A}_n-A_n=o_{\P}(1)$. 

Let us also define $\hat{B}_n:=\sqrt{T/n}\hat{C}_T\hat{E}_T$ if $T/n=O(1)$ and $\hat{B}_n:=0$ if $T/n=o(1)$. 
Once again, using Lemma \ref{conver}, it is easy to show that $\hat{B}_n-B_n=o_{\P}(1)$ 
where $B_n:=\sqrt{T/n}C_T E_T$ if $T/n=O(1)$ and $B_n:=0$ if $T/n=o(1)$. 
Setting $W_n=\sqrt{nT}\left(\hat{\theta}-\theta\right)$, 
we already know from Theorem \ref{NA} that $A_n W_n-B_n\Rightarrow  \mathcal{N}_2(0,I_2)$. 
Moreover, using the decomposition
\begin{eqnarray*}
Z_n&:=&\hat{A}_n W_n-\hat{B}_n\\
&=& A_n W_n-B_n\\
&+&\left(\hat{A}_n-A_n\right)A_n^{-1}\left(A_n W_n-B_n\right)+\left(\hat{A}_nA_n^{-1}-I_2\right)B_n+B_n-\hat{B}_n.
\end{eqnarray*}
We conclude from Lemma \ref{conver} that $Z_n\Rightarrow \mathcal{N}_2(0,I_2)$. 
The result is now a consequence of Lemma \ref{convuni2}. \end{proof}

\section{Appendix}\label{sec9}

This Appendix states and proves several technical lemmas used to prove our main results.

Lemma \ref{prems} below shows that Assumption {\bf A3} leads to a lower bound 
for the population mean and the population variance. 
It also shows that small values of the sample mean and sample variance occur with very small probability if $n$ is large, as noted at the beginning of the proofs section \ref{sec8}.

\begin{lem}\label{prems}
Suppose that Assumption {\bf A3} holds,
{
i.e.,  that there exists some $\delta>0$ such that 
$$p(\delta):=\inf_{t\geq 1}\P\left(X_{1,t}>\delta\right)>0  \quad \mathrm{and}\quad 
q(\delta)=\inf_{t\geq 1}\P\left(\left\vert X_{1,t}-X_{2,t}\right\vert> \delta\right)>0.$$
Then}
\begin{enumerate}
\item
There exists $s>0$ such that $\inf_{t\geq 1}\mu_t\geq s$ and $\inf_{t\geq 1}\sigma_t^2\geq s$.
\item
Let $0<\epsilon<\min\left(\delta p(\delta),\delta^2 q(\delta)/4\right)$.
Then there exists $L>0$ such that 
$$\sup_{t\geq 1}\max\left\{\P(\hat{\mu}_t\leq \epsilon),\P(\hat{\sigma}^2_t\leq \epsilon)\right\}\leq \exp(-Ln).$$
\end{enumerate}
\end{lem}

\begin{proof}
To prove point 1, we use the inequality $x^i\geq \delta^i\mathds{1}_{\{x\geq \delta\}}$ for any $x\geq 0$ and $i=1,2$, which leads to the lower bounds 
\begin{eqnarray*}   
\mu_t&\geq& \delta\P(X_{1,t}>\delta)\geq \delta p(\delta),\\
\sigma_t^2&=&\frac{1}{2}\E(\vert X_{1,t}-X_{2,t}\vert^2)\geq \frac{\delta^2}{2}\P\left(\vert X_{1,t}-X_{2,t}\vert\geq \delta\right)\geq \frac{\delta^2}{2}q(\delta).
\end{eqnarray*}
To prove point 2, we observe that 
$$\hat{\mu}_t\geq \frac{\delta}{n}\sum_{j=1}^n \mathds{1}_{X_{j,t}>\delta}.$$
Then, setting $Z_{j,t}:=\mathds{1}_{X_{j,t}>\delta}-\P(X_{j,t}>\delta)$,
$$\P\left(\hat{\mu}_t\leq \epsilon\right)\leq \P\left(\frac{1}{n}\sum_{j=1}^n Z_{j,t}\leq \frac{\epsilon}{\delta}-p(\delta)\right).$$
Using Hoeffding's exponential inequality \citep[inequality $(2.6)$]{Hoeffding},
we conclude that $\sup_{t\geq 1}\P(\hat{\mu}_t\leq \epsilon)\leq \exp(-nL')$ for some $L'>0$. 
$L'$ can be taken here as $\left(\epsilon/\delta-p(\delta)\right)^2$. 
To control the probability of small values of the empirical variance, we use 
an immediate consequence of Lagrange's identity:
$$\hat{\sigma}^2_t=\frac{1}{2n^2}\sum_{1\leq i\neq j\leq n}\left(X_{i,t}-X_{j,t}\right)^2.$$
This leads to the lower bound
$$\hat{\sigma}^2_t\geq \frac{\delta^2}{2n^2}\sum_{1\leq i\neq j\leq n}\mathds{1}_{\vert X_{i,t}-X_{j,t}\vert>\delta}.$$
Next, set $Z_{i,j}^{(t)}:=\mathds{1}_{\vert X_{i,t}-X_{j,t}\vert>\delta}-\P(\vert X_{i,t}-X_{j,t}\vert>\delta)$.
Then 
$$\P\left(\hat{\sigma}_t^2\leq \epsilon\right)\leq\P\left(\frac{1}{n(n-1)}\sum_{1\leq i\neq j\leq n}Z^{(t)}_{i,j}\leq  \frac{2n^2\epsilon}{\delta^2n(n-1)}-q(\delta)\right).$$
Since $2n^2/(n(n-1))\leq 4$ for every $n\geq 2$, an application of Hoeffding's inequality for $U$-statistics 
\citep[bound $(5.7)$]{Hoeffding}
leads to 
$$\P(\hat{\sigma}_t^2\leq \epsilon)\leq \exp(-2[n/2] t^2)$$
where $[n/2]$ denotes the integer-part of $n/2$ and $t=q(\delta)-\frac{4\epsilon}{\delta^2}$.
It is then always possible to get the conclusion of the Lemma for a suitably chosen constant $L\in (0,L')$. 
\end{proof}

\begin{lem}\label{intern}
{
Suppose that Assumption {\bf A1} holds true and for any $j\geq 1$, the variables $X_{j,1},\ldots,X_{j,T}$ are mutually independent.} Let $\{J_t: t\in \N\}$ be a family of $2\times 2$ matrices such that $\sup_{t\geq 1}\vert J_t\vert<\infty$. 
Suppose that {
$\sup_{t\geq 1} \E X_{1,t}^{4(1+\kappa)}<\infty$ for some $\kappa\in (0,1)$}. Then 
$$\frac{1}{T}\sum_{t=1}^T(\hat{m}_t-m_t)'J_t(\hat{m}_t-m_t)=
\frac{1}{nT}\sum_{t=1}^T\mbox{Tr}\left(J_t\v\left(Y_{1,t}\right)\right)+
O_{\P}\left({
\left(\frac{T}{n}\right)^{\frac{1}{2}-\frac{\kappa}{1+\kappa}}\frac{1}{n^{\frac{2\kappa}{1+\kappa}}\sqrt{nT}}+\frac{1}{\sqrt{n}\sqrt{nT}}}\right),$$
where $Y_{j,t}=(X_{j,t},X_{j,t}^2)'$ for $1\leq j\leq n$ and $1\leq t\leq T$.
\end{lem}

\begin{proof}
Setting $\overline{Y}_{j,t}:=Y_{j,t}-\E Y_{j,t}$ and 
$$A:=\frac{1}{n^2 T}\sum_{t=1}^T\sum_{j=1}^n\overline{Y}_{j,t}'J_t \overline{Y}_{j,t},$$
$$B:=\frac{2}{n^2 T}\sum_{t=1}^T\sum_{1\leq i<j\leq n}\overline{Y}_{i,t}'J_t \overline{Y}_{j,t},$$
we have 
$$\frac{1}{T}\sum_{t=1}^T(\hat{m}_t-m_t)'J_t(\hat{m}_t-m_t)=A+B.$$
It is straightforward to get 
$$\E(A)=\frac{1}{nT}\sum_{t=1}^T\mbox{Tr}\left(J_t\v\left(Y_{1,t}\right)\right).$$
Moreover, {
setting $p=1+\kappa$, one can apply Lemma \ref{moment} to get
$$\Vert A-\E(A)\Vert_p=O\left(n^{-1-\frac{\kappa}{1+\kappa}}T^{-\frac{\kappa}{1+\kappa}}\right).$$}
Then
\begin{eqnarray*}
\E(B^2)&=&\frac{4}{n^4 T^2}\sum_{t=1}^T\sum_{1\leq i<j\leq n}\E(\vert \overline{Y}_{i,t}'J_t\overline{Y}_{j,t}\vert^2)\\
&\leq& \frac{2n(n-1)}{n^4 T}\sup_{t\geq 1}\vert J_t\vert\times \sup_{t\geq 1}\E^2(\vert \overline{Y}_{1,t}\vert^2).
\end{eqnarray*}
Hence $B=O_{\P}\left(\frac{1}{
{\sqrt{n}}\sqrt{nT}}\right)$.
\end{proof}

\begin{lem}\label{convuni}
Assume that $(Z_n)_{n\in \N}$ is a sequence of random vectors of $\R^2$ such that $Z_n\Rightarrow \mathcal{N}\left(0,I_2\right)$. 
Then for any sequence $(u_n)_{n\in\N}$ of deterministic vectors of $\R^2$ such that $\vert u_n\vert=1$, we have 
$$u_n'Z_n\Rightarrow \mathcal{N}(0,1).$$
\end{lem}

\begin{proof}
It is clear that the sequence $\left(u_n' Z_n\right)_{n\in\N}$ is bounded in probability (i.e., tight). 
We simply have to show that any subsequence has a subsubsequence converging to the distribution $\mathcal{N}(0,1)$.
Let $\left(u_{\phi_n}'Z_{\phi_n}\right)_{n\in\N}$ be a given subsequence. Since the unit sphere is compact, one can find 
a subsubsequence $\left(u_{\psi_n}\right)_{n\in\N}$ with $\psi_n\in \left\{\phi_k:k\in\N\right\}$ for all $n\in\N$ and such that
$\lim_{n\rightarrow \infty}u_{\psi_n}=u$ for some $u\in\R^2$ such that $\vert u\vert=1$. Since
$$u_{\psi_n}'Z_{\psi_n}=\left(u_{\psi_n}-u\right)'Z_{\psi_n}+u'Z_{\psi_n},$$
our assumption guarantees that $u' Z_{\psi_n}\Rightarrow \mathcal{N}(0,1)$ and we also know that $\left(u_{\psi_n}-u\right)'Z_{\psi_n}$ will converge in probability to $0$ as $n\rightarrow \infty$. From Slutsky's lemma, we conclude that $u_{\psi_n}'Z_{\psi_n}\Rightarrow \mathcal{N}(0,1)$.
\end{proof}

\begin{lem}\label{convuni2}
Suppose $Z_n=\hat{A}_n W_n-\hat{B}_n\Rightarrow \mathcal{N}_2(0,I_2)$ with $\hat{A}_n$ a random real matrix of size $2\times 2$, 
$W_n$ a random real vector of $\R^2$ and $\hat{B}_n$ a random real vector of $\R^2$. 
Suppose furthermore that $\hat{A}_n-A_n=o_{\P}(1)$, with $A_n$ deterministic and such that $A_n=O(1)$ and $A_n^{-1}=O(1)$.
Setting $H_n:=\hat{A}_n^{-1}(\hat{A}_n^{-1})'$ and $m_n=\hat{A}_n^{-1}\hat{B}_n$, we have for $i=1,2$, 
$$\frac{W_n(i)-m_n(i)}{\sqrt{H_n(i,i)}}\Rightarrow \mathcal{N}(0,1).$$
\end{lem}

\begin{proof}
To show first that 
\begin{equation}\label{inverse}
\hat{A}_n^{-1}-A_n^{-1}=o_{\P}(1),
\end{equation}
we write
$$\hat{A}_n^{-1}=\left(A_n-(A_n-\hat{A}_n)\right)^{-1}=\sum_{k\geq 0}\left(A_n^{-1}(A_n-\hat{A}_n)\right)^kA_n^{-1},$$
which is bounded in probability because $\hat{A}_n-A_n\to 0$ in probability.
We then deduce that 
$$\hat{A}_n^{-1}-A_n^{-1}=\hat{A}_n^{-1}\left(A_n-\hat{A}_n\right)A_n^{-1}=o_{P}(1),$$
which proves (\ref{inverse}). 
Next, for $i=1,2$, we set 
$\hat{v}_{n,i}:=(\hat{A}_n^{-1})'e_i$, where $e_1:=(1,0)', e_2:=(0,1)'$.
Define also $v_{n,i}:=(A_n^{-1})'e_i, i = 1, 2$. 
Since $e_i=A_n'v_{n,i}$ and $A_n=O(1)$, $v_{n,i}$ cannot approximate $0$. Moreover, $\hat{v}_{n,i}-v_{n,i}=o_{\P}(1)$. 
Using also that $A_n^{-1}=O(1)$, {$v_{n,i}$ is bounded and} we deduce that 
$$\frac{\hat{v}_{n,i}}{\left\vert \hat{v}_{n,i}\right\vert}-\frac{v_{n,i}}{\left\vert v_{n,i}\right\vert}=o_{P}(1).$$
We then conclude that 
$$\frac{\hat{v}_{n,i}'Z_n}{\left\vert \hat{v}_{n,i}\right\vert}-\frac{v_{n,i}'Z_n}{\left\vert v_{n,i}\right\vert}=o_{P}(1).$$
Using Lemma \ref{convuni}, we know that ${v_{n,i}
{'}Z_n/\vert v_{n,i}\vert}\Rightarrow \mathcal{N}(0,1)$.
Since 
$$\frac{\hat{v}_{n,i}'Z_n}{\left\vert \hat{v}_{n,i}\right\vert}=\frac{W_n(i)-m_n(i)}{\vert \hat{v}_{n,i}\vert},$$
we only have to show that $\vert v_{n,i}\vert=\sqrt{H_n(i,i)}$ to finish the proof. 
\end{proof}

The following lemma is useful to bound moments of partial sums. It is a very special case of the Burkholder inequality for martingales
\citep[Theorem 2.10]{hall2014martingale}.
\begin{lem}\label{moment}
If $\xi_1,\ldots,\xi_n$ are {
independent} random variables with $0$ mean and a finite moment of order $p>1$, then there exists a constant $C_p>0$ such that
$$\Big\Vert \sum_{i=1}^n \xi_i\Big\Vert_p\leq C_p\times \left\{\begin{array}{c} n^{1/p}\max_{1\leq i\leq n}\Vert \xi_i\Vert_p\mbox{ if } p\in (1,2),\\
n^{1/2}\max_{1\leq i\leq n}\Vert \xi_i\Vert_{p/2}^{1/2}\mbox{ if } p\geq 2.\end{array}\right.$$
\end{lem}

\begin{cor}\label{momcons}
{
Suppose that Assumption {\bf A1} is fulfilled and that} for some $k>4$, $\sup_{t\geq 1}\E\left\vert X_{1,t}\right\vert^k<\infty$. 
Define for $i=1,2,3,4$, $S_{i,t}:=\frac{1}{n}\sum_{j=1}^n (X_{j,t}^i-\E[X_{j,t}^i])$. 
Suppose that for some integer $d\geq 1$ and $i=1,\ldots,4$, $\kappa_i$ is a positive real number, $\alpha_i\in\{1,\ldots,4\}$ with 
$\sum_{i=1}^d\alpha_i\kappa_i\leq k$. 

Then 
$$\sup_{t\geq 1}\E\left[\prod_{i=1}^d \left\vert S_{\alpha_i,t}\right\vert^{\kappa_i}\right]=O\left(n^{-\sum_{i=1}^d\kappa_i\min\{1/2,1-\alpha_i/k\}}\right).$$
\end{cor}

\begin{proof}
Setting 
$p_i:={k/(\alpha_i\kappa_i)}$ for $1\leq i\leq d$, we have $\sum_{i=1}^d 1/p_i\leq 1$ and applying the H\"older inequality, we get 
$$\E\left[\prod_{i=1}^d \left\vert S_{\alpha_i,t}\right\vert^{\kappa_i}\right]\leq \prod_{i=1}^d \E^{\frac{1}{p_i}}\left[\left\vert S_{\alpha_i,t}\right\vert^{\frac{k}{\alpha_i}}\right].$$
The result then follows from Lemma \ref{moment}.
\end{proof}

\begin{lem}\label{defpos}
Given the assumptions of Theorem \ref{NA},
let $\Gamma_T$ be the symmetric, non-negative definite matrix 
$$\Gamma_T:=\frac{1}{T}\sum_{t=1}^T J_{\phi_{\theta}}(m_t)\v\left(Y_{1,t}\right)J_{\phi_{\theta}}(m_t)'.$$
Let $\lambda_T$ be the smallest eigenvalue of $\Gamma_T$ and
$$A_T=\frac{1}{T}\sum_{t=1}^T J_{\phi_{\theta}}(m_t).$$
\begin{enumerate}
\item There exists $\xi>0$ such that
$$\lambda_T\geq \xi\min_{1\leq t\leq T}\det[\v(Y_{1,t})]{\det}^2(A_T).$$
\item
Suppose that $X_{1,t}$ follows a Poisson distribution with parameter $\mu_t=\mu(t/T)$ with $\mu:[0,1]\rightarrow (0,\infty)$ a continuous function.
Then there exists $\zeta>0$ such that for $T$ large enough, $\lambda_T\geq \zeta$.
\end{enumerate}
\end{lem}

\begin{proof}
\begin{enumerate}
\item
Let $\gamma_t$ be the smallest eigenvalue of $\v(Y_{1,t})$ and $\beta_T$ the smallest eigenvalue of the matrix $A_TA_T'$ .
If $x\in\R^2$, then
\begin{eqnarray*}
x'\Gamma_Tx&\geq &\frac{1}{T}\sum_{t=1}^T x'J_{\phi_{\theta}}(m_t)\v(Y_{1,t})J_{\phi_{\theta}}(m_t)'x\\
&\geq& \frac{1}{T}\sum_{t=1}^T\gamma_t \vert J_{\phi_{\theta}}(m_t)'x\vert^2\\
&\geq& \min_{1\leq t\leq T}\gamma_t (x'A_T)^2\\
&\geq& \min_{1\leq t\leq T} \gamma_t \beta_T\\
&\geq& \min_{1\leq t\leq T}\frac{\det\v(Y_{1,t})}{\mbox{Tr} \v(Y_{1,t})}\cdot\frac{\det^2 A_T}{\sum_{1\leq i,j\leq 2}\vert A_T(i,j)\vert^2}.
\end{eqnarray*}
Let us comment on the previous inequalities. 
The second and the fourth one are based on the Courant-Fischer theorem. 
The third one is based on Jensen's inequality.
The last inequality is a consequence of the inequality $\lambda\geq \det(B)/\mbox{Tr}(B)$ for any (non-null) non-negative definite matrix $B$ with an eigenvalue $\lambda$. 
The result announced follows from the Courant-Fischer theorem and the fact that 
the assumptions of Theorem \ref{NA} guarantee that the coefficients of the matrices $A_T$ and $\v(Y_{1,t})$ 
are bounded with respect to $T$ and $t$, respectively.
\item
Since $\mu_t=\sigma_t^2$ for the Poisson distribution, 
using the expressions of the first four moments of the Poisson distribution, 
we get $\det\v(Y_{1,t})=2\mu_t^3$.
The assumption on the function $\mu$ entails that $\inf_{t\geq 1}\mu_t>0$.
Moreover,
$$A_T=\begin{pmatrix} -2-\frac{1}{T}\sum_{t=1}^T \frac{1}{\mu_t}&\frac{1}{T}\sum_{t=1}^T \frac{1}{\mu_t}\\
-\frac{1}{T}\sum_{t=1}^T\left(2+\frac{1}{\mu_t}\right)\log\mu_t& \frac{1}{T}\sum_{t=1}^T \frac{\log \mu_t}{\mu_t}\end{pmatrix}.$$ 
Then 
$$\lim_{T\rightarrow\infty}\det(A_T)=2\int_0^1\frac{1}{\mu(u)}du\cdot\int_0^1\log\mu(u)du-2\int_0^1 \frac{\log\mu(u)}{\mu(u)}du.$$
One can show that the previous limit is positive, since from Jensen's inequality applied first to the probability measure $P(du)=\mathds{1}_{[0,1]}(u)du/\left(\mu(u)\int_0^1\mu(v)^{-1}dv\right)$ and then to the uniform distribution over $[0,1]$,
$$\left(\int_0^1\frac{1}{\mu(u)}du\right)^{-1}\int_0^1 \frac{\log\mu(u)}{\mu(u)}du< -\log \int_0^1 \frac{1}{\mu(u)}du<\int_0^1\log\mu(u)du.$$
\end{enumerate}
This proves the lemma.
\end{proof}

\begin{lem}\label{conver}
Suppose that all the assumptions of Theorem \ref{NA} are valid, with or without the additional condition on $T/n$ given in point $1$ or point $2$. 
Set $\Sigma_t:=\v(Y_{1,t})$ and 
$$\hat{\Sigma}_t:=\frac{1}{n}\sum_{j=1}^n Y_{j,t}Y_{j,t}'-\overline{Y}_{n,t}\overline{Y}_{n,t}',\quad \overline{Y}_{n,t}:=\frac{1}{n}\sum_{j=1}^nY_{j,t},$$
$$\hat{\Gamma}_T:=\frac{1}{T}\sum_{t=1}^T J_{\phi_{\hat{\theta}}}(\hat{m}_t)\hat{\Sigma}_tJ_{\phi_{\hat{\theta}}}(\hat{m}_t)',$$
$$\hat{G}_t:=\left(\mbox{Tr}(\nabla^{(2)}\phi_{1,\hat{\theta}}(\hat{m}_t)\hat{\Sigma}_t),\mbox{Tr}(\nabla^{(2)}\phi_{2,\hat{\theta}}(\hat{m}_t)\hat{\Sigma}_t)\right).$$
Then, with $\Gamma_T$ defined in the statement of Lemma \ref{defpos},
\begin{equation}\label{avantdernier}
\hat{\Gamma}_T-\Gamma_T=o_{\P}(1).
\end{equation}
Moreover, if $C_T:=\Gamma_T^{-1/2}$ and $\hat{C}_T:=\hat{\Gamma}_T^{-1/2}$, we have 
\begin{eqnarray}
\hat{C}_T-C_T&=o_{\P}(1),\label{presque}\\
\frac{1}{T}\sum_{t=1}^T(\hat{G}_t-G_t)&=o_{\P}(1).\label{dernier}
\end{eqnarray}
\end{lem}

\begin{proof}
Under the assumptions of Theorem \ref{NA}, $\phi_{\theta}$ and its Jacobian and Hessian matrices are bounded along the sequence $(m_t)_{t\geq 1}$.

We first prove (\ref{avantdernier}).
From Lemma \ref{prems}, it is only necessary to show the convergence on the set
$$E_n(\epsilon)=\cap_{t=1}^T\left\{\hat{\mu}_t>\epsilon,\hat{\sigma}^2_t>\epsilon\right\}.$$
For such $\epsilon>0$ and any $\kappa\in (0,1)$, there exists $d_{\epsilon,\kappa}>0$ such that for any $z,z'\geq \epsilon$,
$$\left\vert \log z-\log z'\right\vert\leq d_{\epsilon,\kappa}\vert z-z'\vert^{\kappa},\quad \left\vert \frac{1}{z}-\frac{1}{z'}\right\vert\leq d_{\epsilon,\kappa}\vert z-z'\vert^{\kappa}.$$
Let $0<\kappa\leq 1$ to be fixed later. 
Using the previous inequalities, 
one can easily find a constant $c_{\epsilon,\kappa}>0$ 
to satisfy the following inequalities:
\begin{eqnarray*}
\left\vert \hat{\sigma}^2_t-\sigma^2_t\right\vert^{\kappa}&\leq& c_{\epsilon,\kappa}\left[\left\vert \hat{m}_{2,t}-m_{2,t}\right\vert^{\kappa}+\left\vert \hat{\mu}_t-\mu_t\right\vert^{\kappa}+\left\vert \hat{\mu}_t-\mu_t\right\vert^{2\kappa}\right],\\
\left\vert \phi_{1,\hat{\theta}}(\hat{m}_t)-\phi_{1,\theta}(m_t)\right\vert&\leq& c_{\epsilon,\kappa}\left[\left\vert \hat{\sigma}^2_t-\sigma^2_t\right\vert^{\kappa}+\left\vert\hat{\theta}_1-\theta_1\right\vert+\left\vert\hat{\theta}_2-\theta_2\right\vert+{
\left\{\left\vert \hat{\theta}_2-\theta_2\right\vert+1\right\}}\cdot\left\vert\hat{\mu}_t-\mu_t\right\vert^{\kappa}\right],\\
\left\vert \frac{\partial\phi_{1,\hat{\theta}}}{\partial x}(\hat{m}_t)-\frac{\partial\phi_{1,\theta}}{\partial x}(m_t)\right\vert&\leq& c_{\epsilon,\kappa}\left[\left\vert \hat{\sigma}^2_t-\sigma^2_t\right\vert^{\kappa}+\left\vert\hat{\mu}_t-\mu_t\right\vert+\left\vert\hat{\theta}_2-\theta_2\right\vert+\left\vert\hat{\mu}_t-\mu_t \right\vert\cdot\left\vert\hat{\sigma}^2_t-\sigma^2_t\right\vert^{\kappa}\right],\\
\left\vert \frac{\partial\phi_{1,\hat{\theta}}}{\partial y}(\hat{m}_t)-\frac{\partial\phi_{1,\theta}}{\partial y}(m_t)\right\vert&\leq& c_{\epsilon,\kappa}\left\vert \hat{\sigma}^2_t-\sigma^2_t\right\vert^{\kappa},\\
\left\vert \frac{\partial\phi_{2,\hat{\theta}}}{\partial x}(\hat{m}_t)-\frac{\partial\phi_{2,\theta}}{\partial x}(m_t)\right\vert&\leq& c_{\epsilon,\kappa}\left[\left\vert \hat{\mu}_t-\mu_t\right\vert^{\kappa}+\left\vert \phi_{1,\hat{\theta}}(\hat{m}_t)-\phi_{1,\theta}(m_t)\right\vert\cdot\left\vert\hat{\mu}_t-\mu_t\right\vert^{\kappa}\right]\\
&+&c_{\epsilon,\kappa}\left[\left\vert\phi_{1,\hat{\theta}}(\hat{m}_t)-\phi_{1,\theta}(m_t)\right\vert+\left\vert\hat{\mu}_t-\mu_t \right\vert^{\kappa}\cdot\left\vert \frac{\partial\phi_{1,\hat{\theta}}}{\partial x}(\hat{m}_t)-\frac{\partial\phi_{1,\theta}}{\partial x}(m_t)\right\vert\right]\\
&+&c_{\epsilon,\kappa}\left\vert \frac{\partial\phi_{1,\hat{\theta}}}{\partial x}(\hat{m}_t)-\frac{\partial\phi_{1,\theta}}{\partial x}(m_t)\right\vert,\\
\left\vert \frac{\partial\phi_{2,\hat{\theta}}}{\partial y}(\hat{m}_t)-\frac{\partial\phi_{2,\theta}}{\partial y}(m_t)\right\vert&\leq& c_{\epsilon,\kappa}\left[\left\vert \hat{\mu}_t-\mu_t\right\vert^{\kappa}+
\left\vert \hat{\mu}_t-\mu_t\right\vert^{\kappa}\cdot \left\vert \frac{\partial\phi_{1,\hat{\theta}}}{\partial y}(\hat{m}_t)-\frac{\partial\phi_{1,\theta}}{\partial y}(m_t)\right\vert\right]\\
&+& c_{\epsilon,\kappa}\left\vert \frac{\partial\phi_{1,\hat{\theta}}}{\partial y}(\hat{m}_t)-\frac{\partial\phi_{1,\theta}}{\partial y}(m_t)\right\vert.
\end{eqnarray*}
Moreover, there exists a constant $c>0$ such that
\begin{eqnarray*}
\vert \hat{\Sigma}_t(1,2)-\Sigma_t(1,2)\vert&\leq &c[\vert \hat{m}_{3,t}-m_{3,t}\vert+\vert\hat{m}_{1,t}-m_{1,t}\vert\cdot\vert \hat{m}_{2,t}-m_{2,t}\vert+\vert\hat{m}_{1,t}-m_{1,t}\vert+\vert\hat{m}_{2,t}-m_{2,t}\vert],\\
\vert \hat{\Sigma}_t(2,2)-\Sigma_t(2,2)\vert&\leq& c[\vert \hat{m}_{4,t}-m_{4,t}\vert+\vert\hat{m}_{2,t}-m_{2,t}\vert^2+\vert \hat{m}_{2,t}-m_{2,t}\vert].
\end{eqnarray*}
Next, setting $\hat{A}_t=J_{\phi_{\hat{\theta}}}(\hat{m}_t)$ and $A_t=J_{\phi_{\theta}}(m_t)$, we use the decomposition
\begin{eqnarray*}
\hat{A}_t\hat{\Sigma}\hat{A}_t-A_t\Sigma_tA_t&=&A_t\Sigma_t(\hat{A}_t-A_t)+A_t(\hat{\Sigma}_t-\Sigma_t)A_t+A_t(\hat{\Sigma}_t-\Sigma_t)(\hat{A}_t-A_t)\\
&+& (\hat{A}_t-A_t)\Sigma_t A_t+(\hat{A}_t-A_t)\Sigma_t(\hat{A}_t-A_t)+(\hat{A}_t-A_t)(\hat{\Sigma}_t-\Sigma_t)A_t\\
&+&(\hat{A}_t-A_t)(\hat{\Sigma}_t-\Sigma_t)(\hat{A}_t-A_t).
\end{eqnarray*}
We deduce from the previous bounds that the entries of the difference {
$\hat{\Gamma}_T-\Gamma_T$}
can be decomposed as a sum of terms, each of which is bounded 
up to a constant possibly depending on $\vert\hat{\theta}_i-\theta_i\vert$, $i=1,2$,  
by a partial sum of the form $\frac{1}{T}\sum_{t=1}^T\prod_{i=1}^d \vert S_{\alpha_i,t}\vert^{\kappa_i}$
for some integer $d\geq 1$ and for $i=1,\ldots,4$, $\kappa_i$ is a positive real number, $\alpha_i\in\{1,\ldots,4\}$ with $\sum_{i=1}^d\alpha_i\kappa_i\leq 4+6\kappa$. Choosing $\kappa\in (0,1)$ such that 
$4+6\kappa\leq {
k}$, one can apply Corollary \ref{momcons} and use the weak consistency of $\hat{\theta}$ to conclude that $\hat{\Gamma}_T-\Gamma_T=o_{\P}(1)$.

We now prove (\ref{presque}). First we show that 
\begin{equation}\label{mimic}
\hat{\Gamma}_T^{-1}-\Gamma_T^{-1}=o_{\P}(1)
\end{equation}
using the decomposition 
$$\hat{\Gamma}_T^{-1}-\Gamma_T^{-1}=\hat{\Gamma}_T^{-1}(\Gamma_T-\hat{\Gamma}_T)\Gamma_T^{-1}.$$
From our assumptions, there exists $K>0$ such that for $n$ large enough, $\vert\Gamma_T^{-1}\vert\leq K$. 
It is then only necessary to show that 
\begin{equation}\label{mimic2}
\hat{\Gamma}_T^{-1}=O_{\P}(1).
\end{equation}
Let $0<\eta<1/K$ and set $M:=
{K}/(1-\eta K)$. We have 
$$\P(\vert \hat{\Gamma}_T^{-1}\vert>M)\leq \P\left(\displaystyle{
\sup}_{\vert H\vert\leq \eta}\vert (\Gamma_T-H)^{-1}\vert>M\right)+\P(\vert\hat{\Gamma}_T-\Gamma_T\vert>\eta).$$
The second probability on the right side goes to $0$ as $n\rightarrow \infty$ and the first one is 0 because 
$$(\Gamma_T-H)^{-1}=\sum_{k\geq 0}\left(\Gamma_T^{-1}H\right)^k\Gamma_T^{-1}$$
has a norm smaller than $M$. This shows (\ref{mimic2}) and then (\ref{mimic}).

Next, using the Ando-Hemmen inequality \citep[Proposition  3.2]{Ando},
 we have for any unitarily invariant norm $\vert\cdot\vert$,
$$\left\vert \hat{\Gamma}_T^{-1/2}-\Gamma_T^{-1/2}\right\vert\leq \frac{\left\vert \hat{\Gamma}_T^{-1}-\Gamma_T^{-1}\right\vert}{\sqrt{\lambda_{-}\left(\hat{\Gamma}_T^{-1}\right)}+\sqrt{\lambda_{-}\left(\Gamma_T^{-1}\right)}}\leq \frac{\left\vert \hat{\Gamma}_T^{-1}-\Gamma_T^{-1}\right\vert}{\sqrt{\lambda_{-}\left(\Gamma_T^{-1}\right)}},$$
where $\lambda_{-}(B)$ denotes the smallest eigenvalue of a symmetric matrix $B$. Of course, the previous inequality is valid for an arbitrary norm $\vert\cdot\vert$, up to a constant factor, by equivalence of the norms 
{
in a finite dimensional space}.
If we denote by $\lambda_{+}(B)$ the largest eigenvalue of a matrix $B$, we have
$$\varliminf\lambda_{-}\left(\Gamma_T^{-1}\right)=1/\varlimsup\lambda_{+}\left(\Gamma_T\right)>0$$
and $\hat{C}_T-C_T=o_{\P}(1)$. 

To prove (\ref{dernier}),
we proceed as for the proof of (\ref{avantdernier}).
We first bound the differences 
$$\left\vert \frac{\partial^2\phi_{\ell,\hat{\theta}}}{\partial x^i\partial y^j}\left(\hat{m}_t\right)-\frac{\partial^2\phi_{\ell,\theta}}{\partial x^i\partial y^j}\left(m_t\right)\right\vert, \quad\ell=1,2,\quad i,j\in \{0,1,2\} \text{ such that }i+j=2.$$
We find that each difference can be bounded by a sum of terms of the form $$C_1\vert \hat{\theta}_i-\theta_i\vert, \qquad \mbox{or} \qquad
C_2\prod_{k=1}^d \vert S_{\alpha_k,t}\vert^{\kappa_k}$$ 
with $C_1,C_2>0$, $\alpha_k\in \{1,2\}$, $\kappa_k\geq 0$ and $\sum_{k=1}^d\alpha_k\kappa_k$ 
bounded by $2+2\kappa$ if $(\ell,i,j)=(1,2,0)$, $1+2\kappa$ if $(\ell,i,j)=(1,1,1)$, $2\kappa$ if $(\ell,i,j)=(1,0,2)$, $2+3\kappa$ if $(\ell,i,j)=(2,2,0)$, 
$1+3\kappa$ for $(\ell,i,j)=(2,1,1)$ and $3\kappa$ if $(\ell,i,j)=(2,0,2)$. 
It is then easily seen that the differences
$\frac{1}{T}\sum_{t=1}^T\left(\hat{G}_t-G_t\right)$ only involve some factors of the previous type with $\sum_{k=1}^d\alpha_k\kappa_k\leq 4+3\kappa$ where $\kappa>0$ can be arbitrarily small. We then get the result from Corollary \ref{momcons}, choosing $\kappa$ small enough such that $4+3\kappa\leq {
k}$.
\end{proof}

\section{Acknowledgements}

This work was funded by  CY-AS
("Investissements d'Avenir" ANR-16-IDEX-0008), "EcoDep" PSI-AAP2020-0000000013.
P.D.  was also funded by the chair FIME: {\tt https://fime-lab.org/}.
J.E.C. thanks Roseanne Benjamin for her help during this work.

\bibliographystyle{plainnat}
\bibliography{bibTaylor20230921}

\end{document}